\pdfoutput=1
\RequirePackage{ifpdf}
\ifpdf 
\documentclass[pdftex]{sigma}
\else
\documentclass{sigma}
\fi

\newcommand{\corr}[1]{\langle {#1} \rangle}
\newcommand{\Corr}[1]{\biggl\langle {#1} \biggr\rangle}

\newcommand{\e}{\epsilon}

   \newcommand{\bs}{{\bf s}}
  \newcommand{\bp}{{\bf p}}

 \newcommand{\bP}{\mathbb{P}}

 \newcommand{\F}{\mathcal{F}}
\newcommand{\ZZ}{\mathbb{Z}}
 
\newcommand{\pd}{\partial}
\newcommand{\p}{\partial}
\newcommand{\tr}{{\rm tr}}
\newcommand{\nn}{\nonumber}
\newcommand{\Mbar}{\overline{\mathcal M}}
\newcommand{\half}{\frac{1}{2}}

\newcommand{\be}{\begin{equation}}
\newcommand{\beq}{\begin{equation}}
\newcommand{\ee}{\end{equation}}
\newcommand{\eeq}{\end{equation}}
\newcommand{\bea}{\begin{eqnarray}}
\newcommand{\eea}{\end{eqnarray}}
\newcommand{\ben}{\begin{eqnarray*}}
\newcommand{\een}{\end{eqnarray*}}

\numberwithin{equation}{section}

\newtheorem{Theorem}{Theorem}[section]
\newtheorem*{Theorem*}{Theorem}
\newtheorem{Corollary}[Theorem]{Corollary}
\newtheorem{Lemma}[Theorem]{Lemma}
\newtheorem{Proposition}[Theorem]{Proposition}
 { \theoremstyle{definition}

\newtheorem{Remark}[Theorem]{Remark} }

\begin{document}

\newcommand{\arXivNumber}{2311.06506}

\renewcommand{\PaperNumber}{068}

\FirstPageHeading

\ShortArticleName{From Toda Hierarchy to KP Hierarchy}

\ArticleName{From Toda Hierarchy to KP Hierarchy}

\Author{Di YANG~$^{\rm a}$ and Jian ZHOU~$^{\rm b}$}

\AuthorNameForHeading{D.~Yang and J.~Zhou}

\Address{$^{\rm a)}$~School of Mathematical Sciences, University of Science and Technology of China,\\
\hphantom{$^{\rm a)}$}~Hefei 230026, P.R.~China}
\EmailD{\href{mailto:diyang@ustc.edu.cn}{diyang@ustc.edu.cn}}

\Address{$^{\rm b)}$~Department of Mathematical Sciences, Tsinghua University,
Beijing 100084, P.R.~China}
\EmailD{\href{mailto:jianzhou@mail.tsinghua.edu.cn}{jianzhou@mail.tsinghua.edu.cn}}

\ArticleDates{Received October 08, 2024, in final form July 27, 2025; Published online August 09, 2025}

\Abstract{Using the matrix-resolvent method and a formula of the second-named author on the $n$-point function for a KP tau-function, we show that the tau-function of an arbitrary solution to the Toda lattice hierarchy is a KP tau-function. We then generalize this result to tau-functions for the extended Toda hierarchy (ETH) by developing the matrix-resolvent method for the ETH. As an example the partition function of Gromov--Witten invariants of the complex projective line is a KP tau-function, and an application on irreducible representations of the symmetric group is obtained.}

\Keywords{Toda hierarchy; KP hierarchy; matrix-resolvent method; complex projective line; Gromov--Witten invariant}

\Classification{37K10; 05E05; 14N35; 53D45; 05E10}

\section{Introduction}
The Kadomtsev--Petviashvili (KP) hierarchy and
the Toda lattice hierarchy are two important integrable hierarchies. In this
paper, we will show that the tau-function of an arbitrary solution to the Toda lattice hierarchy
gives an infinite family of tau-functions of the KP hierarchy.
This is achieved by combining two different results:
a formula~\cite{Zhou-Emergent} of the second-named author on the $n$-point function for a KP tau-function in the big cell,
and a work~\cite{Y} of the first-named author on the $n$-point function for a Toda tau-function.

The KP hierarchy is an infinite family of equations with infinitely many unknown functions,
which can be written using the Lax pair formalism as follows:
\beq\label{KPhierarchy30}
\frac{\p L_{\rm KP}}{\p T_k} = \bigl[\bigl(L_{\rm KP}^k\bigr)_+, L_{\rm KP}\bigr], \qquad k\ge1,
\eeq
where
\[
L_{\rm KP} = \p + \sum_{j\ge1} u_j \p^{-j}, \qquad \p=\p/\p X,
\]
is a pseudodifferential operator, called the {\it Lax operator}.
For details about pseudodifferential operators and their operations see for example~\cite{Dickey}.
The independent variables $T_1,T_2,T_3,\dots$ are called {\it times}.
Since $\p L_{\rm KP}/\p T_1 = \p L_{\rm KP}/\p X$, we identify $T_1$ with $X$.

Let ${\bf T}=(T_1,T_2,\dots)$ denote the infinite vector of KP times.
It is well known (see, e.g.,~\cite{Dickey}) that an arbitrary solution $(u_1({\bf T}), u_2({\bf T}),\dots)$
to the KP hierarchy~\eqref{KPhierarchy30} can be compactly represented by a single function $\tau=\tau({\bf T})$
called the {\it tau-function} as
\[
u_1({\bf T}) = \frac{\p^2 \log \tau}{\p T_1^2}, \qquad u_2({\bf T}) = \frac12 \frac{\p^2 \log \tau}{\p T_1 \p T_2} - \frac12 \frac{\p^3 \log \tau}{\p T_1^3}, \qquad \dots,
\]
and the tau-function~$\tau$ satisfies the {\it Hirota bilinear identities} given by
\beq\label{bilinearidH}
{\rm res}_{\lambda=\infty} \tau\bigl({\bf T} - \bigl[\lambda^{-1}\bigr]\bigr) \tau\bigl({\bf T'} + \bigl[\lambda^{-1}\bigr]\bigr) {\rm e}^{\sum_{k\ge1} (T_k-T_k') \lambda^k} {\rm d}\lambda=0.
\eeq
Here \smash{$\bigl[\lambda^{-1}\bigr]:=\bigl(\frac1\lambda, \frac1{2\lambda^2}, \frac 1{3\lambda^3}, \dots\bigr)$}.
Equivalence between~\eqref{bilinearidH} and~\eqref{KPhierarchy30} is a standard result in the theory of integrable systems; see, for example,~\cite{Dickey}.
Equation~\eqref{bilinearidH} itself gives the defining equations for a KP tau-function (namely, we do not
 need to start with a solution $(u_1({\bf T}), u_2({\bf T}),\dots)$ to the KP hierarchy~\eqref{KPhierarchy30}).

Let $\Lambda\colon f(x)\mapsto f(x+\e)$ be the shift operator and
\begin{gather*}
L := \Lambda + v(x) + w(x) \Lambda^{-1}
\end{gather*}
the {\it Lax operator}. Here $\e$ is a parameter.
The {\it Toda lattice hierarchy} (also known as the 1d Toda chain or 1-Toda hierarchy)
can be defined using the Lax pair formalism as
\beq\label{TLH}
\e\frac{\p L}{\p m_i} = \frac1{(i+1)!} \bigl[\bigl(L^{i+1}\bigr)_+,L\bigr], \qquad i\geq0.
\eeq
Here, for a difference operator~$P$ written in the form $P=\sum_{k\in \ZZ} P_k \Lambda^k$,
$P_+$ is defined as $\sum_{k\geq0} P_k \Lambda^k$. We also denote ${\bf m}=(m_0,m_1,m_2,\dots)$.

It is known~\cite{UT} that an arbitrary tau-function $\tau(N,{\bf x},{\bf y})$ of the 2-Toda hierarchy~\cite{UT}
is a~KP~tau-function with respect to either ${\bf x}$ or~${\bf y}$ as KP times (fixing the others
as parameters). The~following Theorem~\ref{thm:toda-kp},
which could essentially be deduced from~\cite{UT} (see Remark~\ref{remark12} below), gives a similar statement for the 1-Toda hierarchy. Let $(v(x,{\bf m};\e), w(x,{\bf m};\e))$ be an arbitrary power-series-in-${\bf m}$ solution
to the Toda lattice hierarchy
with coefficients being in a ring $V$ of functions of~$x$ closed under shifting~$x$
by $\pm \e$ (e.g., $V$ could be $\mathbb{C}[x,\e]$),
and $\tau(x,{\bf m};\e)$ the tau-function~\cite{DY} of this solution.
\begin{Theorem} \label{thm:toda-kp}
The tau-function $\tau(x,{\bf m};\e)$ for the Toda lattice hierarchy
is a tau-function of the KP hierarchy for any $x$ and $\e$,
where ${\bf m}$ and the KP times ${\bf T}$ are related by
\be\label{tT30}
m_i = (i+1)! \, \e \, T_{i+1},\qquad i\ge0.
\ee
\end{Theorem}
\begin{Remark}\label{remark12}
Theorem~\ref{thm:toda-kp} should be known to experts, and could be deduced
from a claim~\cite[p.~30]{UT} (cf.\ \cite{KMMM}) that the Lax representation of the
Toda lattice hierarchy can be obtained from
a reduction of the 2-Toda hierarchy with
the reduction condition
\[
\p/\p x_i = \p /\p y_i, \qquad i\ge1,
\]
up to a quadratic function
of ${\bf x}$, ${\bf y}$. Here, ${\bf x}=(x_1,x_2,\dots)$, ${\bf y}=(y_1,y_2,\dots)$.
(Our formula~\eqref{dpsiabtoda} below should shed light on the Toda chain discussions in~\cite{KMMM}.)
We would like to thank A.~Alexandrov for bringing our attention to the reduction arguments of~\cite{UT}.
We also point out the following well-known fact: the first nontrivial Hirota bilinear equation
of the 2-Toda hierarchy under the reduction condition coincides with that of the 1-Toda hierarchy. Thus, an alternative way to
show the equivalence between the 1-Toda hierarchy and the reduction of the 2-Toda hierarchy is
to compare the Hirota bilinear equations of the two hierarchies~\cite{Milanov, UT}, whose detail might deserve
a further study.
In this work we take a different approach. For the KP hierarchy,
a formula for the connected $n$-point functions was derived by the second-named author using the fermionic approach in~\cite{Zhou-Emergent},
whereas for the Toda lattice hierarchy a formula of the same form was derived by the second-named author~\cite{Y}, but using a totally different method:
the matrix-resolvent method. So we present a~different proof combining these two approaches. An advantage of this proof is that it can be generalized
to the extended Toda hierarchy~\cite{CDZ} whose fermionic approach has not yet appeared. See our Theorem~\ref{thm:toda-kp-ext} below. The main content of Section~\ref{sectionTodaext}
is to generalize the matrix-resolvent method to the case of the extended Toda hierarchy.
\end{Remark}

Let us now recall some earlier results in the literature that motivate us to obtain this result.
It is well known in matrix model theory~\cite{AvM, GMMMO}
that using the theory of orthogonal polynomials,
the GUE partition function $\{Z_N({\bf T};\e)\}_{N \geq 1}$ is a tau-function for the Toda lattice hierarchy.
Here~$N$ is the size of the Hermitian matrix used to define the partition function, and $x=N\e$.
A~perhaps less well-known result is that for each~$N$, $Z_N({\bf T};\e)$ is also a tau-function of the KP hierarchy~\cite{Shaw-Tu-Yen} (see also~\cite{KMMM}).
Both of these results have been revisited recently and some new perspectives naturally arise.

First of all,
in~\cite[Deﬁnition 1.2.2]{DY},
a notion of matrix resolvent for the Toda lattice hierarchy was introduced.
Based on this,
a definition of the tau-function of Dubrovin--Zhang type of the Toda lattice hierarchy was given in~\cite[Definition 1.2.4]{DY}.
Usually a normalization constant is chosen so that the partition function of Gaussian unitary ensemble (GUE) is equal to $1$
when all the coupling constants are set to be equal to zero.
In ~\cite[Appendix A]{DY}, it was shown that after multiplying by a suitable correction factor,
the GUE partition functions give us a tau-function of Dubrovin--Zhang type.
 In~\cite{DY}, it was also shown that the correction factor can be obtained by using the theory of
Toda tau-functions, so this method is applicable to other examples.

Note the major goal of~\cite{DY} is to develop a method
for computing $n$-point function associated with the tau-function
of the Toda lattice hierarchy based on the matrix-resolvent method,
by generalizing the earlier results developed by Bertola, Dubrovin and the first-named
author in the cases of KdV hierarchy~\cite{BDY1} and Drinfeld--Sokolov hierarchies associated
with simple Lie algebras~\cite{BDY2}.
Earlier the second-named author proved an explicit formula~\cite{Zhou-WK} for the Schur expansion
of the Witten--Kontsevich tau-function of the KdV hierarchy.
Balogh and the first-named author interpreted this formula in terms of the
affine coordinates of the Witten--Kontsevich tau-function in~\cite{BY}.
Inspired by~\cite{BDY1},
the second-named author proved a formula~\cite{Zhou-Emergent}
 of $n$-point function associated with
an arbitrary tau-function of the KP hierarchy,
based on the affine coordinates of the element in the Sato Grassmannian
corresponding to the tau-function.
The formula in~\cite{DY} has a difference from the formula in~\cite{Zhou-Emergent}:
In the former matrices of power series are used whereas ordinary power series are used in the latter.
To remedy the difference, the first-named author proved a formula for $n$-point function
using only power series in~\cite{Y} for tau-functions of Toda lattice hierarchy (see also~\cite{DYZ}).

Secondly,
it was shown in~\cite{Zhou-GUE} that the normalized GUE partition function
\smash{$\frac{Z_N({\bf T};\e)}{Z_N({\bf 0};\e)}$}
gives rise to~a~family of KP tau-functions in the big cell parameterized by the
t'Hooft coupling constant~${x=N\e}$ (see also~\cite{KKN, KMMM, Morozov, Shaw-Tu-Yen}).
Furthermore,
an explicit formula for the affine coordinates for this family of KP tau-function was derived in~\cite{Zhou-GUE}.
As an application,
the formulas for the corresponding $n$-point functions were obtained by applying the formula in~\cite{Zhou-Emergent}.

Now the GUE partition functions can be studied from two different perspectives: either as a tau-function of Dubrovin--Zhang type of the Toda lattice hierarchy,
or as a family of tau-functions of the KP hierarchy.
The belief that this is just a special case of general phenomenon leads us to Theorem \ref{thm:toda-kp}.

There are earlier results that also lead to Theorem~\ref{thm:toda-kp}.
It is well known that by introducing the t'Hooft coupling constant $x= N\e$,
\smash{$\log \frac{Z_N({\bf T};\e)}{Z_N({\bf 0};\e)}$} has an expansion as weighted sum
of ribbon graphs~\cite{DY, HZ, Zhou-GUE}.
In other words, the enumeration of ribbon graphs gives rise to a special family of tau-functions
of the KP hierarchy which become a special tau-function for the Toda lattice hierarchy
by multiplying by a suitable correction factor.
Recall that ribbon graphs can be regarded as {\em clean} dessins.
As a generalization,
the weighted sum of Grothendieck's dessins d'enfants has
been shown by Zograf~\cite{Zograf} (see also Kazarian and Zograf~\cite{KZ})
to be a tau-function of the KP hierarchy, that is
referred to as the dessin partition function.
Based on this fact, the second-named author had found the affine coordinates
of the tau-function of the KP hierarchy associated with the dessin counting~\cite{Zhou-dessins},
and so explicit formula for the $n$-point functions of the dessin partition function
can be written down using the general formula in~\cite{Zhou-Emergent}.
The dessin partition function is a family of KP tau-functions parameterized by two parameters.
It reduces to some other well-known tau-function by suitably specifying these parameters~\cite{Zhou-dessins}.
A proposal to study the dualities among different models based on the theory of KP hierarchy
was then proposed in~\cite{Zhou-dessins} and was further elaborated in~\cite{Zhou-dessins2}.

The study of KP tau-functions and the study of Toda tau-functions were merged
in our earlier work~\cite{YZ} in our study of the dessin tau-function and its role in unifying various theories.
Since counting dessins is similar to counting ribbon graphs and
GUE provides a matrix model for enumerating the ribbon graphs,
one naturally expects a matrix model for enumerating dessins.
The problem was addressed by Ambj{\o}rn and Chekhov~\cite{AC}, where
they proposed a matrix model in the class of generalized Kontsevich models
and an equivalent Hermitian 1-matrix model.
In~\cite{YZ}, we found that the Laguerre unitary ensemble (LUE)
gives a simpler and rigorous matrix model for the dessins
counting in the sense that the dessin partition function is equal to the {\em normalized}
LUE partition function.
Furthermore, by multiplying by a suitable correction factor,
the dessin partition function gives rise to a Dubrovin--Zhang type tau-function of the Toda lattice hierarchy.
It was proposed in~\cite{YZ} to study the
duality between dessin partition function with partition functions of other theories from the
viewpoint of Toda lattice hierarchy.
In this new approach, one can apply the theory of normal forms of integrable hierarchies
and the extended Toda hierarchy~\cite{CDZ} as developed by Dubrovin and Zhang~\cite{DZ-norm, DZ-toda}.
In fact,
the relevant Frobenius manifold is the Frobenius manifold associated with the Gromov--Witten invariants
of~$\bP^1$.
It has already been used in~\cite{Du09} (cf.\ \cite{DY, Y2}) to calculate the GUE correlators.

Based on some earlier results of Dijkgraaf and Witten~\cite{DW},
Eguchi and Yang~\cite{EY} proposed a~matrix model for the Gromov--Witten invariants of $\bP^1$ (see also~\cite{EHY}).
Besides the Toda lattice hierarchy,
the $\bP^1$-partition function satisfies an extra family of flows introduced in~\cite{EY, Getzler, Zhang} (see also~\cite{Carlet, CDZ, Du93, EY, Milanov}).
These flows are called {\it extended flows} by Getzler~\cite{Getzler} and Zhang~\cite{Zhang}.
They can be defined using the Lax pair formalism~\cite{CDZ} as follows:
\beq\label{ext-flows}
\e\frac{\p L}{\p b_i} = \frac{2}{i!}\bigl[\bigl(L^i(\log L - c_i)\bigr)_+,L\bigr], \qquad i\geq0,
\eeq
where $c_i:=\sum_{j=1}^i \frac1j$ are harmonic numbers and for the definition of $\log L$ see~\cite{CDZ, Milanov}.
The flows~\eqref{ext-flows} commute with the traditional flows~\eqref{TLH}
of the Toda lattice hierarchy, and they also pairwise commute.
All-together, \eqref{TLH}, \eqref{ext-flows} form the {\it extended Toda hierarchy} (ETH)~\cite{CDZ, Getzler, Zhang}.

Let $(v({\bf b}+x{\bf 1}, {\bf m}; \e), w({\bf b}+x{\bf 1}, {\bf m}; \e) )$ be an arbitrary solution
to the ETH, and $\tau({\bf b}+x{\bf 1},{\bf m}; \e)$ the tau-function~\cite{CDZ} of this solution
(although it is uniquely determined up to multiplying by
the exponential of an affine-linear function of ${\bf b}, {\bf m}$). 
Here ${\bf b}=(b_0,b_1,b_2,\dots)$ and ${\bf 1}=(1, 0, 0, \dots)$.
The following theorem
generalizes Theorem~\ref{thm:toda-kp}.
\begin{Theorem} \label{thm:toda-kp-ext}
The tau-function $\tau({\bf b}+x{\bf 1},{\bf m};\e)$ for the ETH
is a KP tau-function for any $x$ and ${\bf b}$, where ${\bf m}$ and~${\bf T}$ are related by~\eqref{tT30}.
\end{Theorem}
The proof is in Section~\ref{sectionTodaext}.

\begin{Remark}
For our intended applications to Gromov--Witten theory, we are content with tau-functions as formal power series.
The results of Theorems~\ref{thm:toda-kp} and~\ref{thm:toda-kp-ext} raise the problem to characterize the KP tau-functions obtained from the tau-functions of Toda lattice
hierarchy or the extended Toda hierarchy. This problem is beyond the scope of this work and will be pursued in the future investigation.
In viewpoint of Sato Grassmannian we think that our formula~\eqref{dpsiabtoda} below sheds light on this discussion, which we also
mentioned briefly in Remark~\ref{remark12}.
We thank an anonymous referee for posing this problem to us.
\end{Remark}

It was conjectured in~\cite{Du93, Du96, EY, Getzler}, and was proved
by Dubrovin and Zhang~\cite{DZ-toda} that assuming the validity of the so-called Virasoro constraints (which is
proved, e.g., in~\cite{OP3})
the partition function \smash{$Z^{\bP^1}({\bf b}+x{\bf 1}, {\bf m}; \e)$} of
Gromov--Witten (GW) invariants of~\smash{$\bP^1$} (see Section~\ref{sectionP1} for the definition) is the tau-function
of a particular solution to the ETH. Independently, Okounkov and Pandharipande~\cite{OP2, OP1} proved that
$Z^{\bP^1}(x{\bf 1}, {\bf m}; \e)$ satisfies the bilinear equations for the~1-Toda hierarchy.
Based on the above-mentioned results and on Theorem~\ref{thm:toda-kp-ext}
we immediately obtain the following corollary.

\begin{Corollary} \label{ThmMainP1}
The partition function $Z^{\bP^1}({\bf b}+x{\bf 1} ,{\bf m}; \e)$ of GW invariants of~$\bP^1$ is a KP tau-function
for any $x$, $\e$ and~${\bf b}$,
where ${\bf m}$ is related to~${\bf T}$ by~\eqref{tT30}. In particular, $Z^{\bP^1}(x{\bf 1},{\bf m};\e)$ is a~KP tau-function for any $x$ and~$\e$.
\end{Corollary}

The stationary sector of the $\bP^1$-partition function is $Z^{\bP^1}(0{\bf 1} ,{\bf m};\e)$.
We compute in Section~\ref{sectionP1} its affine coordinates by two different methods: One by the results of~\cite{DYZ},
the other by the results of~\cite{OP1}.
The former gives us an explicit closed formula,
and the latter gives an expression in~terms of characters of the symmetric groups.
We then get a closed formula for some Plancherel averages in Corollary~\ref{cor:Plancherel}.

The rest of the paper is organized as follows.
In Section~\ref{sectionKP}, we recall the formula of $n$-point function for a KP tau-function.
In Section~\ref{sectionToda}, we show that an arbitrary tau-function for the Toda lattice hierarchy is a KP tau-function with the ground ring
being the ring of functions of the space variable of Toda. The generalization to ETH is given in Section~\ref{sectionTodaext}.
In Section~\ref{sectionP1}, we give an application of Corollary~\ref{ThmMainP1}.

\section[The formula for the n-point function for a KP tau-function]{The formula for the $\boldsymbol{n}$-point function for a KP tau-function}\label{sectionKP}

In this section, we first review the explicit formula (see Theorem~\ref{thm:n-point-KP} below) obtained by the second-named author
for the $n$-point function for an arbitrary KP tau-function in the big cell, and then consider the converse statement that
gives a criterion for a KP tau-function.

Before entering into the details let us introduce some {\it notations}.
Denote by $\mathcal{R}$ a suitable ground ring.
For a formal power series $F({\bf p})\in \mathcal{R}[[{\bf p}]]$,
with ${\bf p}=(p_1,p_2,\dots)$, define the $n$-point function associated with $F({\bf p})$ by
\[
G_n(\xi_1, \dots, \xi_n) = \sum_{k_1,\dots,k_n\ge1} \prod_{i=1}^n \frac{k_i}{\xi_i^{k_i+1}} \cdot
\frac{\pd^n F({\bf p})}{\pd p_{k_1} \cdots \pd p_{k_n}}\biggr|_{{\bf p} = {\bf 0}}, \qquad n\ge0.
\]
By a {\it partition} $\mu=(\mu_1,\mu_2,\dots)$, we mean a sequence of
weakly decreasing non-negative integers with $\mu_k=0$
 for sufficiently large $k$.
 The length $\ell(\mu)$ is the number of the non-zero parts of~$\mu$, the weight $|\mu|:=\mu_1+\mu_2+\cdots$,
 and the multiplicity of $i$ in~$\mu$ is denoted by $m_i(\mu)$.
 The set of all partitions will be denoted by~$\mathcal{P}$, and the set of partitions of weight~$d$ is denoted by $\mathcal{P}_d$.
 The
Schur polynomial $s_{\mu}({\bf p})$ associated to~$\mu\in\mathcal{P}$ is a polynomial in the variables
${\bf p} = (p_1,p_2,\dots)$,
defined by
\beq\label{defSchur1}
s_{\mu}({\bf p}) := \det_{1\leq i,j \leq \ell(\mu)}(h_{\mu_i-i+j}({\bf p})),
\eeq
where $h_j({\bf p})$ are polynomials defined by the generating function
\[
\sum_{j=0}^{\infty}h_j({\bf p})z^j := {\rm e}^{\sum_{k=1}^{\infty} \frac{p_k}{k} z^k}.
\]

One constructive way to describe tau-functions of the KP hierarchy is
to use Sato's Grassmannian. Let
\smash{$z^{1/2}\mathcal{R}\bigl[\bigl[z,z^{-1}\bigr]\bigr]$} be the space of formal series in $z$ of half integral powers.
An~element~$V$ in the big cell $Gr_0$ of Sato's Grassmannian is specified by a sequence of series
\begin{gather*} 
\Psi_j(z) = z^{j+1/2} + \sum_{i=0}^\infty A_{i,j} z^{-i-1/2}, \qquad j \geq 0,
\end{gather*}
where the coefficients $A_{i,j}$ are called the {\em affine coordinates} of~$V$. To any element~$V\in Gr_0$, there corresponds
a particular KP tau-function $Z_V$ defined by
\beq\label{fromaffinetotau}
Z_V:= \sum_{\lambda\in\mathcal{P}} c_\lambda s_\lambda({\bf p}),
\eeq
where $T_k=p_k/k$, and for a partition~$\lambda$ written in terms of the Frobenius notation~\cite[Section~I.1, p.~3]{Mac}
$\lambda=(m_1,\dots,m_k|n_1,\dots,n_k)$,
\[
c_\lambda := (-1)^{n_1+\dots+n_k} \det_{1\leq i,j\leq k} \bigl(A_{m_i, n_j}\bigr).
\]
In particular, for $\lambda=\bigl(i+1,1^j\bigr)$ we have $c_\lambda=(-1)^j A_{i,j}$.
So one can easily read off the affine coordinates from the expansion~\eqref{fromaffinetotau}.

The following theorem was obtained in~\cite{Zhou-Emergent}.

\begin{Theorem}[\cite{Zhou-Emergent}]\label{thm:n-point-KP}
The $n$-point function associated with $\log Z_V$ is given by the following formula:
For $n=1$,
\be \label{eqn:G-1}
G_1(\xi) = \sum_{i,j\geq 0} A_{i,j} \xi^{-i-j-2},
\ee
and for $n \ge2 $,
\be \label{eqn:G-n}
G_n(\xi_1, \ldots, \xi_n)
= (-1)^{n-1} \sum_{\text{$n$-cycles}}
\prod_{i=1}^n B\bigl(\xi_{\sigma(i)}, \xi_{\sigma(i+1)}\bigr)
- \frac{\delta_{n,2}}{(\xi_1-\xi_2)^2},
\ee
where $\sigma(n+1)$ is understood as $\sigma(1)$, and
\[
B(\xi_i, \xi_j) = \begin{cases}
\dfrac{1}{\xi_i-\xi_j} + A(\xi_i, \xi_j), & i \neq j, \\
A(\xi_i, \xi_i), & i =j, \\
\end{cases}
\]
and
\[
\begin{split}
A(\xi, \eta)
= & \sum_{i,j \geq 0} A_{i,j} \xi^{-j-1} \eta^{-i-1} .
\end{split}
\]
\end{Theorem}

Theorem~\ref{thm:n-point-KP} enables one to easily compute the $n$-point function once
we have found the affine coordinates $A_{i,j}$.
For example, there is only one $2$-cycle,
so the two-point function is given by
\ben
G_2(\xi_1, \xi_2)
& = & \frac{A(\xi_1, \xi_2) - A(\xi_2,\xi_1)}{\xi_1-\xi_2}
-A(\xi_1, \xi_2)\cdot A(\xi_2,\xi_1).
\een

Let us consider the converse of the above Theorem~\ref{thm:n-point-KP}. Recall the simple fact, which can be verified directly using~\eqref{bilinearidH}, that
multiplying a KP tau-function $\tau({\bf T}({\bf p}))$
by the exponential of an arbitrary affine linear function of~${\bf p}$
\beq\label{linear1004}
\tau({\bf T}({\bf p})) \cdot {\rm e}^{C_0+\sum_{k\ge1} C_k p_k}, \qquad C_k\in\mathcal{R},\quad k\ge0,
\eeq
produces again a KP tau-function. We recall here that $T_k=p_k/k$. The renormalized variables~$p_k$ are also sometimes called KP times.

\begin{Corollary}\label{conversthmA}
Let $F$ be a formal power series in~${\bf p}$ and $G_n(\xi_1,\dots,\xi_n)$ its $n$-point function.
If~there are $(A_{i,j})_{i,j\geq 0}$ such that for all $n\ge2$
$G_n$ are given by~\eqref{eqn:G-n},
then $Z= {\rm e}^F$ is a tau-function of the KP hierarchy.
\end{Corollary}
\begin{proof}
By Theorem~\ref{thm:n-point-KP},
the tau-function corresponds to the point in the big cell with $A_{i,j}$ being the affine coordinates
can only differ from $Z$ by multiplying by the exponential of some affine-linear function.
\end{proof}

\begin{Remark}
From the above proof, we see that Corollary~\ref{conversthmA} directly follows from Theorem~\ref{thm:n-point-KP}. In~\cite{ABDBKS},
Theorem~\ref{thm:n-point-KP} together with Corollary~\ref{conversthmA} is regarded as {\it Zhou's theorem}.
\end{Remark}

The next proposition describes how the affine coordinates change under~\eqref{linear1004}.

\begin{Proposition}\label{VVprimelinear}
Let $V$, $V'$ be two elements in the big cell having affine coordinates~$A_{i,j}$,~$A'_{i,j}$, respectively,
and let $Z_V$, $Z_{V'}$ be the KP tau-functions associated to~$V$, $V'$. If there exist $C_1,C_2,\dots$ such that
$
Z_{V'} = {\rm e}^{\sum_{k\ge1} C_k p_k} Z_V$,
then we have the following identity:
\begin{gather*}
\sum_{i,j \geq 0} A'_{i,j} \xi^{-j-1} \eta^{-i-1} + \frac1{\xi-\eta} = \biggl(\sum_{i,j \geq 0} A_{i,j} \xi^{-j-1} \eta^{-i-1} + \frac1{\xi-\eta}\biggr)
\frac{{\rm e}^{-\sum_{\ell\ge1} C_\ell \xi^{-\ell}}}{{\rm e}^{-\sum_{\ell\ge1} C_\ell \eta^{-\ell}}}.
\end{gather*}
\end{Proposition}
\begin{proof}
By direct verifications using~\eqref{eqn:G-1}.
\end{proof}

\section{From the Toda lattice hierarchy to the KP hierarchy}\label{sectionToda}
In this section, we show that an arbitrary tau-function for the Toda lattice hierarchy is a KP tau-function.

Let $(v(x,{\bf m};\e), w(x,{\bf m};\e))$ be an arbitrary solution in $W[[{\bf m}]]^2$
to the Toda lattice hierarchy~\eqref{TLH}, and $\tau(x,{\bf m};\e)$ the tau-function~\cite{CDZ, DY, DZ-toda} of this solution.
Here, $W$ is a certain ring of functions of~$x$ closed under shifting~$x$
by $\pm \e$, and $\tau(x,{\bf m};\e)$ lives in \smash{$\widetilde W[[{\bf m}]]$} with \smash{$\widetilde W$} being some extension of~$W$.
The solution $(v(x,{\bf m};\e), w(x,{\bf m};\e))$ is in one-to-one correspondence
with its initial value
\[
v(x,{\bf 0};\e)=:f(x,\e), \qquad w(x,{\bf 0}; \e)=:g(x,\e).
\]
Denote by $L_{\rm ini}$ the initial Lax operator
\[
L_{\rm ini}=\Lambda+f(x,\e) + g(x,\e)\Lambda^{-1},
\]
and let
\[
s(x,\e) = - \bigl(1-\Lambda^{-1}\bigr)^{-1} (\log g(x,\e)),
\]
which lives in the extended ring~\smash{$\widetilde{W}$}.
Recall from~\cite{Y} that two elements
\[
 \psi_1(\lambda,x;\e) = \bigl(1+{\rm O}\bigl(\lambda^{-1}\bigr)\bigr)\lambda^{x/\e}, \qquad
\psi_2(\lambda,x;\e)=\bigl(1+{\rm O}\bigl(\lambda^{-1}\bigr)\bigr){\rm e}^{-s(x,\e)}\lambda^{-x/\e}
\]
are called forming {\it a pair of wave functions} of $L_{\rm ini}$ if
\begin{align*}
& L_{\rm ini} (\psi_1) = \lambda \psi_1, \qquad L_{\rm ini} (\psi_2) = \lambda \psi_2,\\
& d(\lambda,x;\e):= \psi_1 \, \Lambda^{-1}(\psi_2) - \psi_2 \, \Lambda^{-1}(\psi_1) = \lambda {\rm e}^{-s(x-\e,\e)}.
\end{align*}
The following formula is proved in~\cite{Y}:
\begin{align}
& \e^n \sum_{i_1,\dots,i_n\ge0}
\frac{(i_1+1)! \cdots (i_n+1)!} {\lambda_1^{i_1+2}\cdots \lambda_n^{i_n+2}} \,
\frac{\p^n \log \tau(x,{\bf m};\e)}{\p m_{i_1} \cdots \p m_{i_n}}\bigg|_{{\bf m}={\bf 0}}\nn\\
&\qquad{}= (-1)^{n-1} \frac{{\rm e}^{n s(x-\e,\e)}}{\prod_{j=1}^n \lambda_j} \sum_{\text{$n$-cycles}}
\prod_{j=1}^n D\bigl(\lambda_{\sigma(j)},\lambda_{\sigma(j+1)}; x;\e\bigr) - \frac{\delta_{n,2}}{(\lambda-\mu)^2},\label{Y421}
\end{align}
where $n\ge2$, and
\beq\label{dpsiabtoda}
D(\lambda,\mu; x;\e) := \frac{\psi_1(\lambda,x;\e) \psi_2(\mu,x-\e;\e) - \psi_1(\lambda,x-\e;\e) \psi_2(\mu,x;\e)}{\lambda-\mu}.
\eeq
\begin{Remark}
The function $D(\lambda,\mu; x;\e)$ looks similar to the
Christoffel--Darboux kernel in~matrix models~\cite{Deift, Mehta}.
We hope to study their relations in a future work.
\end{Remark}

It was observed in~\cite[(99)]{Y} that
\smash{$\frac{{\rm e}^{s(x-\e,\e)}}{\mu} D(\lambda,\mu; x;\e) \frac{\mu^{x/\e}}{\lambda^{x/\e}} - \frac1{\lambda-\mu}$}
is a power series of~$\lambda^{-1}$,~$\mu^{-1}$.
By a more careful analysis we can show that
\beq\label{eqn:hat-A-26}
\frac{{\rm e}^{s(x-\e,\e)}}{\mu} D(\lambda,\mu; x;\e) \frac{\mu^{\frac{x}\e}}{\lambda^{\frac{x}\e}} =
\frac1{\lambda-\mu} + \sum_{i,j\ge0} \frac{A_{i,j}(x,\e)}{\lambda^{j+1} \mu^{i+1}} =: B(\lambda,\mu;x;\e)
\eeq
for some coefficients $A_{i,j}(x,\e)$.
In the next section, we will give a complete proof of a generalized version of~\eqref{eqn:hat-A-26}.

Using~\eqref{eqn:hat-A-26} and~\eqref{Y421} we find
\begin{align}
&\e^n \sum_{i_1,\dots,i_n\ge0}
\frac{(i_1+1)! \cdots (i_n+1)!} {\lambda_1^{i_1+2}\cdots \lambda_n^{i_n+2}} \,
\frac{\p^n \log \tau(x,{\bf m};\e)}{\p m_{i_1} \cdots \p m_{i_n}}\bigg|_{{\bf m}={\bf 0}}\nn \\
&\qquad= (-1)^{n-1} \sum_{\text{$n$-cycles}}
\prod_{j=1}^n B\bigl(\lambda_{\sigma(j)},\lambda_{\sigma(j+1)};\e\bigr) - \frac{\delta_{n,2}}{(\lambda-\mu)^2}, \qquad n\ge2.\label{Y926}
\end{align}
(See also~\cite[Corollary~1]{Y}.)

We are ready to prove Theorem~\ref{thm:toda-kp}.

\begin{proof}[Proof of Theorem~\ref{thm:toda-kp}]
By using \eqref{eqn:hat-A-26}, \eqref{Y926} and~\eqref{tT30}, we know that the $n$-point function for~${n\ge2}$
associated to $\log\tau(x,{\bf m};\e)$ has the form~\eqref{eqn:G-n}.
The theorem is then proved by applying Corollary~\ref{conversthmA}.
\end{proof}

If we write
\[
\tau(x,{\bf m};\e) = \tau_{\rm corr}(x,\e) \tau_1 (x,{\bf m};\e), \qquad \tau_1(x, {\bf m}={\bf 0};\e)\equiv1,
\]
then by Theorem~\ref{thm:toda-kp} we know that $\tau_1(x, {\bf m}; \e)$ is a KP tau-function in the big cell.
The factor~$\tau_{\rm corr}(x, \e)$ can~\cite{CDZ, Y} be determined by
\beq\label{correctionfactor}
\frac{\tau_{\rm corr}(x+\e,\e)\tau_{\rm corr}(x-\e,\e)}{\tau_{\rm corr}(x,\e)^2} = g(x,\e) = {\rm e}^{s(x-\e,\e)-s(x,\e)}.
\eeq
This factor can be identified with the {\it correction factor} in~\cite{DY, Y, YZ}.

Let $\tilde{\tau}_1(x, {\bf m};\e)$ be the KP tau-function
associated to the point in Sato's Grassmannian having the affine coordinates
$A_{i,j}(x,\e)$, with $A_{i,j}(x,\e)$ given in~\eqref{eqn:hat-A-26}. Then there exist
$a_i(x,\e)$, $i\ge0$, such that
\[
\tau_1(x, {\bf m};\e) = {\rm e}^{\sum_{i\ge0} a_i(x,\e) m_i} \tilde{\tau}_1(x, {\bf m};\e).
\]
By further applying Proposition~\ref{VVprimelinear}, we get the affine coordinates for $\tau_1(x, {\bf m};\e)$.

\section{Extended Toda flows in the KP hierarchy}\label{sectionTodaext}
The goal of this section is to prove Theorem~\ref{thm:toda-kp-ext}.

Before entering into the main construction for this section, we briefly recall here the definition of $\log L$.
It is shown in~\cite{CDZ} that there exist dressing operators
\begin{align*}
P= \sum_{k\ge0} P_k \Lambda^{-k}, \qquad P_0= 1,\qquad
Q= \sum_{k\ge0} Q_k \Lambda^k, 
\end{align*}
such that
\[
L=P\circ\Lambda\circ P^{-1}=Q\circ\Lambda^{-1}\circ Q^{-1}.
\]
Here, $P_k$, $Q_k$ belong to a certain extension of the differential polynomial ring.
The logarithm of the Lax operator is then defined by~\cite{CDZ}
\[
\log L= \frac12 P\circ \e\p_x \circ P^{-1} - \frac12 Q\circ \e\p_x \circ Q^{-1}.
\]

As in the introduction, for an arbitrary solution $(v({\bf b}+x{\bf 1}, {\bf m};\e), w({\bf b}+x{\bf 1}, {\bf m};\e))$ to the ETH,
let $\tau({\bf b}+x{\bf 1}, {\bf m};\e)$ be the tau-function (again in the sense of~\cite{CDZ, DZ-toda})
 of the solution. Let us also denote
 \beq\label{defsigma}
 \sigma(x, {\bf b}, {\bf m};\e) = \log\frac{\tau({\bf b}+x{\bf 1}, {\bf m};\e)}{\tau({\bf b}+(x+\e){\bf 1}, {\bf m};\e)}.
 \eeq

\subsection{The matrix-resolvent method}
The matrix-resolvent method for studying tau-functions for the Toda lattice hierarchy was developed in~\cite{DY} (see also~\cite{Y}).
By the locality nature of this method, the same formulation as~for the Toda lattice hierarchy
applies to the ETH.

Denote by
\[
\mathcal{A}=\ZZ[v_0, w_0, v_{\pm1}, w_{\pm 1}, v_{\pm 2}, w_{\pm 2}, \dots]
\]
the polynomial ring.
Recall that the {\it basic matrix resolvent} $R(\lambda)$ is defined~\cite{DY} as
the unique element in ${\rm Mat}\bigl(2,\mathcal{A}\bigl[\bigl[\lambda^{-1}\bigr]\bigr]\bigr)$ satisfying
\begin{align*}
& \Lambda(R(\lambda)) U(\lambda) - U(\lambda) R(\lambda) =0,\\
& \tr \, R(\lambda) =1, \qquad \det R(\lambda)=0,\qquad
 R(\lambda) - \begin{pmatrix} 1 & 0 \\ 0 & 0 \\ \end{pmatrix} \in {\rm Mat}\bigl(2,\mathcal{A}\bigl[\bigl[\lambda^{-1}\bigr]\bigr]\lambda^{-1}\bigr), 
\end{align*}
where
\[
U(\lambda) := \begin{pmatrix} v_0-\lambda & w_0 \\ -1 & 0 \\ \end{pmatrix}.
\]
Define \smash{$S_i=\frac1{(i+1)!}{\rm Coef} \bigl(\Lambda(R(\lambda)_{21}), \lambda^{-i-2}\bigr)\in\mathcal{A}$}, $i\geq0$, and
a sequence of elements $\Omega_{i,j} \in \mathcal{A}$ by
\[
\e^2\sum_{i,j\geq0} \frac{(i+1)!(j+1)!}{\lambda^{i+2} \mu^{j+2}}\Omega_{i,j} = \frac{\tr \, R(\lambda)R(\mu)}{(\lambda-\mu)^2} - \frac{1}{(\lambda-\mu)^2}.
\]

According to~\cite{DY}, the above-defined $(\Omega_{i,j}, S_i)$ gives rise to part of
the canonical tau-structure for the ETH in~\cite{CDZ, DZ-toda}.
Indeed, it is shown in~\cite{DY} that $(\Omega_{i,j}, S_i)$ is associated to the
tau-symmetric hamiltonian densities $h_{\alpha,p}$, $\alpha=1,2$, $p\ge-1$,~\cite{CDZ, DZ-toda} for the Toda lattice hierarchy.
More precisely, $S_i=h_{2,i-1}$ and the following identities hold
for an arbitrary solution $(v({\bf b}+x{\bf 1},{\bf m};\e), w({\bf b}+x{\bf 1},{\bf m};\e))$ to the ETH:
\begin{align}
& \e^2\frac{\p^2 \log\tau({\bf b}+x{\bf 1},{\bf m};\e)}{\p m_i \p m_j}=\Omega_{i,j}({\bf b}+x{\bf 1},{\bf m};\e),\nn \\
& \e \frac{\p}{\p m_i} \biggl(\log \frac{\tau({\bf b}+(x+\e){\bf 1},{\bf m};\e)}{\tau({\bf b}+x{\bf 1},{\bf m};\e)}\biggr) =
S_i({\bf b}+x{\bf 1},{\bf m};\e), \label{deftau2} \\
& \frac{\tau({\bf b}+(x+\e){\bf 1},{\bf m};\e)\tau({\bf b}+(x-\e){\bf 1},{\bf m};\e)}{\tau({\bf b}+x{\bf 1},{\bf m};\e)^2} = w({\bf b}+x{\bf 1},{\bf m};\e). \label{deftau3}
\end{align}
Here, $\Omega_{i,j}({\bf b}+x{\bf 1},{\bf m};\e)$ and
$S_i({\bf b}+x{\bf 1},{\bf m};\e)$ are obtained by replacing $v_k$, $w_k$ by $v({\bf b}+x{\bf 1}+k\e,{\bf m};\e)$ and
$w({\bf b}+x{\bf 1}+k\e,{\bf m};\e)$, respectively.
According to~\cite{DY}, for any $n\ge2$,
\begin{align}
&\e^n \sum_{i_1,\dots,i_n\ge0}
\prod_{j=1}^n \frac{(i_j+1)!}{\lambda_j^{i_j+2}} \frac{\p^n \log \tau({\bf b}+x{\bf 1},{\bf m}; \e)}{\p m_{i_1} \cdots \p m_{i_n}}\nn\\
&\qquad{}= - \sum_{\text{$n$-cycles}}
\frac{\tr \, \prod_{j=1}^n R(\lambda_{\sigma(j)})}{\prod_{j=1}^n \bigl(\lambda_{\sigma(j)}-\lambda_{\sigma(j+1)}\bigr)} - \frac{\delta_{n,2}}{(\lambda-\mu)^2}.\label{multiderivativesresolvent}
\end{align}

Note that by~\eqref{defsigma} and \eqref{deftau2}, \eqref{deftau3}, we have
\begin{align}
& - \e\frac{\p \sigma(x, {\bf b},{\bf m};\e)}{\p m_i} = S_i({\bf b}+x{\bf 1},{\bf m};\e), \label{property1sigma} \\
& {\rm e}^{\sigma(x-\e, {\bf b},{\bf m};\e)-\sigma(x, {\bf b},{\bf m};\e)} = w({\bf b}+x{\bf 1},{\bf m};\e). \label{propertysigma}
\end{align}
It is also clear from~\eqref{correctionfactor} that we can choose
$\sigma(x, b_0=0, b_1=0, \dots, {\bf m}={\bf 0};\e)=s(x,\e)$.

\subsection{Wave functions for ETH and the KP hierarchy}
It is shown in~\cite{Carlet} (see also~\cite{CL, Milanov}) that there exist
dressing operators $P$, $Q$ of the forms
\begin{align}
&P= \sum_{k\ge0} P_k(x, {\bf b}, {\bf m};\e) \Lambda^{-k}, \qquad P_0(x, {\bf b}, {\bf m};\e)\equiv 1, \label{CoePk1005}\\
&Q= \sum_{k\ge0} Q_k(x, {\bf b}, {\bf m};\e) \Lambda^k, \label{CoeQk1005}
\end{align}
such that
\begin{align*}
& L=P\circ\Lambda\circ P^{-1}=Q\circ\Lambda^{-1}\circ Q^{-1}, \\
& \e\frac{\p P}{\p m_i} = - \frac1{(i+1)!} \bigl(L^{i+1}\bigr)_- \circ P, \\
& \e\frac{\p P}{\p b_i} = - \frac{2}{i!}\bigl(L^{i}(\log L-c_i)\bigr)_- \circ P,\\
& \e\frac{\p Q}{\p m_i} = \frac1{(i+1)!}\bigl(L^{i+1}\bigr)_+ \circ Q, \\
& \e\frac{\p Q}{\p b_i} = \frac{2}{i!}\bigl(L^{i}(\log L-c_i)\bigr)_+ \circ Q.
\end{align*}
Define
\begin{align*}
& \psi_1(\lambda; x, {\bf b}, {\bf m};\e) :=
P\Bigl(\lambda^{\frac{x}\e} {\rm e}^{\sum_{i\ge0} \frac{2}{i!} \frac{b_i}\e \lambda^i (\log \lambda-c_i)+\sum_{i\ge0} \frac{m_i}{(i+1)!\e} \lambda^{i+1}}\Bigr) , \\
& \psi_2(\lambda; x, {\bf b}, {\bf m};\e) := Q\Bigl(\lambda^{-\frac{x}\e}{\rm e}^{-\sum_{i\ge0} \frac{2}{i!} \frac{b_i}\e \lambda^i (\log \lambda-c_i)-\sum_{i\ge0} \frac{m_i}{(i+1)!\e} \lambda^{i+1}}\Bigr).
\end{align*}
Then $\psi_a=\psi_a(\lambda; x, {\bf b}, {\bf m};\e)$, $a=1,2$, satisfy the following equations:
\begin{align}
& L\psi_a = \lambda \psi_a, \qquad a=1,2, \label{Laxsp1005}\\
& \e \frac{\p \psi_1}{\p m_i} = \frac1{(i+1)!} \bigl(L^{i+1}\bigr)_+ \psi_1,\nn \\
& \e \frac{\p \psi_1}{\p b_i} = \frac{2}{i!} L^i(\log L - c_i)_+ \psi_1, \label{Laxa1bi} \\
& \e \frac{\p \psi_2}{\p m_i} = -\frac1{(i+1)!} \bigl(L^{i+1}\bigr)_- \psi_2,\nn \\
& \e \frac{\p \psi_2}{\p b_i} = -\frac{2}{i!} L^i(\log L - c_i)_- \psi_2. \label{Laxa2bi}
\end{align}
Moreover, $\psi_1$, $\psi_2$ have the form
\begin{align}
&\psi_1 = \phi_1(\lambda; x, {\bf b}, {\bf m};\e) \lambda^{\frac{x}\e}
 {\rm e}^{\sum_{i\ge0} \frac{2}{i!} \frac{b_i}{\e} \lambda^i (\log \lambda-c_i)+\sum_{i\ge0} \frac{m_i}{(i+1)! \e} \lambda^{i+1}} , \label{psiphiform1} \\
&\psi_2 = \phi_2(\lambda; x, {\bf b}, {\bf m};\e) {\rm e}^{-\sigma(x,{\bf b},{\bf m};\e)} \lambda^{-\frac{x}\e}
 {\rm e}^{-\sum_{i\ge0} \frac{2}{i!} \frac{b_i}{\e} \lambda^i (\log \lambda-c_i)-\sum_{i\ge0} \frac{m_i}{(i+1)! \e} \lambda^{i+1}} , \label{psiphiform2}
\end{align}
where
\begin{align}
& \phi_1(\lambda; x, {\bf b}, {\bf m};\e) =\sum_{k\ge0} P_k(x, {\bf b}, {\bf m};\e) \lambda^{-k}, \label{phiexpansion1}\\
& \phi_2(\lambda; x, {\bf b}, {\bf m};\e) =\sum_{k\ge0} B_k(x, {\bf b}, {\bf m};\e) \lambda^{-k}. \label{phiexpansion2}
\end{align}
We recall that $P_k(x, {\bf b}, {\bf m};\e)$ are given in~\eqref{CoePk1005},
and we also note using \eqref{CoeQk1005} that $B_k(x, {\bf b}, {\bf m};\e)=Q_k(x, {\bf b},{\bf m};\e) {\rm e}^{\sigma(x, {\bf b},{\bf m};\e)}$, $k\geq 0$,
and $B_0(x, {\bf b}, {\bf m};\e)\equiv 1$.
We call $\psi_1$ {\it the wave function of type~A} and $\psi_2$ {\it the wave function of type~B},
{\it associated to the solution $(v(x, {\bf b}, {\bf m};\e), w(x, {\bf b}, {\bf m};\e))$}.

Denote
\begin{align*}
&d(\lambda; x, {\bf b}, {\bf m};\e) \\
&\qquad{}:= \psi_1(\lambda; x, {\bf b}, {\bf m};\e) \psi_2(\lambda; x-\e, {\bf b}, {\bf m};\e)
- \psi_1(\lambda; x-\e, {\bf b}, {\bf m};\e) \psi_2(\lambda; x, {\bf b}, {\bf m};\e).
\end{align*}
The following lemma is important.
\begin{Lemma}\label{keylemma}
For any $i\ge0$, we have the identities
\begin{align}
& \frac{\p \bigl({\rm e}^{\Lambda^{-1}(\sigma(x,{\bf b},{\bf m};\e))}d(\lambda; x, {\bf b}, {\bf m};\e)\bigr)}{\p b_i} =0 , \label{id1} \\
& \frac{\p \bigl({\rm e}^{\Lambda^{-1}(\sigma(x,{\bf b},{\bf m};\e))}d(\lambda; x, {\bf b}, {\bf m};\e)\bigr)}{\p m_i} =0. \label{id2}
\end{align}
\end{Lemma}
\begin{proof}
We first introduce some notations,
\begin{align*}
& R_j := {\rm res} L^j, \qquad N_j={\rm Coef}\bigl(L^j, \Lambda^{-1}\bigr), \\
& r_j:= {\rm res} \log L \circ L^j, \qquad n_j= {\rm Coef}\bigl(\log L \circ L^j, \Lambda^{-1}\bigr).
\end{align*}
({\it Warning: avoid from confusing the notation $R_j$ with the notation for the matrix resolvent.})
For example~ (see \cite{Carlet}),
\begin{align*}
& R_0=1, \qquad R_1=v, \qquad N_0=0, \qquad N_1=w, \\
& r_0= \frac12 \Lambda \circ \frac{\e\p}{\Lambda-1} (v), \qquad n_0 = \frac12 \frac{\e\p}{\Lambda-1} (\log w) .
\end{align*}
Here we omitted the arguments in $R_j$, $r_j$, $N_j$, $n_j$, $v$, $w$. Below, when
no ambiguity will occur, we often do this type of omission.

As it was given in~\cite{CDZ}, the tau-symmetric hamiltonian densities $h_{\alpha, p}$ for the ETH
are related to $R_j$, $r_j$ by
\begin{align}
& h_{2, p} = \frac1{(p+2)!} R_{p+2}, \label{CDZtausymm2}\\
& h_{1,p} = \frac2{(p+1)!} r_{p+1} - \frac{2c_{p+1}}{(p+1)!} R_{p+1}, \qquad p\ge-1.\nn
\end{align}
We also recall from~\cite{CDZ} the following formula:
\[
h_{\alpha,i-1} = \e(\Lambda-1) \frac{\p \log \tau}{\p t^{\alpha,i}} = - \e\frac{\p \sigma}{\p t^{\alpha,i}}, \qquad i\ge0.
\]

Let us now prove~\eqref{id2}. (Although \eqref{id2} was already proved in~\cite{Y},
our proof here will be slightly different from~\cite{Y} and can be generalized to proving~\eqref{id1}.)
We have
\beq\label{derivativemi1005}
{\rm e}^{-\Lambda^{-1}(\sigma)}\frac{\p \bigl({\rm e}^{\Lambda^{-1}(\sigma)}d\bigr)}{\p m_i}
= \frac{\p \Lambda^{-1}(\sigma)}{\p m_i} d + \frac{\p d}{\p m_i}, \qquad i\ge0.
\eeq
Denote
\[
B_i:= -\bigl(L^{i+1}\bigr)_-, \qquad i\ge-1.
\]
Then for all $i\ge0$, we have
\begin{align*}
B_i & = -\bigl(L^i\circ L\bigr)_- = - \bigl(\bigl(L^i\bigr)_+ + \bigl(L^i\bigr)_-\bigr) \circ L_-
 = - (R_i L)_- - \bigl(\bigl(L^i\bigr)_- \circ L\bigr)_- \nn\\
& = B_{i-1} \circ L + N_i - R_i w \Lambda^{-1} = \cdots
 = \sum_{j=0}^i \bigl(N_j-R_j w\Lambda^{-1}\bigr) \circ L^{i-j}. \nn
\end{align*}
(The resulting equality is also valid for $i=-1$ as $B_{-1}=0$.)
Thus,
\beq\label{psi2mi1005}
(i+1)! \e \frac{\p \psi_2}{\p m_i} = B_i (\psi_2) = \sum_{j=0}^i \lambda^{i-j} \bigl(N_j-R_j w\Lambda^{-1}\bigr) (\psi_2),
\eeq
and
\beq\label{psi1mi1005}
(i+1)! \e \frac{\p \psi_1}{\p m_i} 
=\lambda^{i+1} \psi_1
+ \sum_{j=0}^i \lambda^{i-j} \bigl(N_j-R_j w\Lambda^{-1}\bigr) (\psi_1).
\eeq
By substituting \eqref{psi1mi1005} and \eqref{psi2mi1005} in~\eqref{derivativemi1005}
and by a lengthy but straightforward calculation using also~\eqref{Laxsp1005} and \eqref{property1sigma}, we obtain
\begin{align*}
{\rm e}^{-\Lambda^{-1}(\sigma)} \e\frac{\p \bigl({\rm e}^{\Lambda^{-1}(\sigma)}d\bigr)}{\p m_i}
={}& {-}{} \Lambda^{-1}(S_{i}) d
+ \frac{\Lambda^{-1}(R_{i+1})}{(i+1)!} d \nn\\
&{}{+}{} \frac{d}{(i+1)!}\sum_{j=0}^i \lambda^{i-j} \bigl(N_j+\Lambda^{-1}(N_j)+\Lambda^{-1}(v R_j-R_{j+1})\bigr) = 0. \nn
\end{align*}
Note that in the last equality we used \eqref{deftau2}, \eqref{CDZtausymm2}
and the relation $L^{k+1}=L\circ L^k=L^k \circ L$ (which implies the vanishing of each
summand in the above $\sum_j$).

We proceed to prove~\eqref{id1}. Like in~\cite{Y}, we note that
\begin{align*}
&\Lambda \bigl({\rm e}^{\Lambda^{-1}(\sigma(x,{\bf b},{\bf m};\e))}d(\lambda; x,{\bf b}, {\bf m};\e)\bigr)
 = {\rm e}^{\sigma} (\Lambda(\psi_1) \psi_2
- \psi_1 \Lambda( \psi_2) ) \nn\\
&\qquad{} = {\rm e}^{\sigma} \bigl\{ \bigl((\lambda-v) \psi_1 - w \Lambda^{-1}(\psi_1)\bigr) \psi_2 - \psi_1 \bigl((\lambda-v) \psi_2 - w \Lambda^{-1}(\psi_2)\bigr) \bigr\} \nn\\
&\qquad{} = {\rm e}^{\Lambda^{-1}(\sigma(x, {\bf b},{\bf m};\e))}d(\lambda; x,{\bf b}, {\bf m};\e) \nn
\end{align*}
(the last equality used~\eqref{propertysigma}).
It follows that
\[
\e\p_x \bigl({\rm e}^{\Lambda^{-1}(\sigma(x,{\bf b},{\bf m};\e))}d(\lambda; x, {\bf b}, {\bf m};\e)\bigr)
= \frac{\e\p_x}{\Lambda-1} \circ (\Lambda-1) \bigl({\rm e}^{\Lambda^{-1}(\sigma(x, {\bf b},{\bf m};\e))}
d(\lambda; x, {\bf b}, {\bf m};\e)\bigr) = 0.
\]
Thus,
\begin{align}
0 & = {\rm e}^{-\Lambda^{-1}(\sigma(x,{\bf b},{\bf m};\e))} \p_x \bigl({\rm e}^{\Lambda^{-1}(\sigma(x,{\bf b},{\bf m};\e))}d(\lambda; x, {\bf b}, {\bf m};\e)\bigr) \nn\\
& = - \Lambda^{-1}(h_{1,-1}) + \p_x(d(\lambda; x, {\bf b}, {\bf m};\e)) = - 2\Lambda^{-1}(r_0)d + \p_x(d(\lambda; x, {\bf b}, {\bf m};\e))\nn.
\end{align}
Denote
\[
Q_i = -2 \bigl(L^i(\log L-c_i)\bigr)_- = -2 \bigl(L^i \log L\bigr)_0-2c_i B_{i-1}, \qquad i\ge0.
\]
In the way similar to the derivation for $B_i$, we find
\[
Q_i= \sum_{j=0}^{i-1} \bigl(2m_j-2r_jw\Lambda^{-1}\bigr) \circ L^{i-j} - 2 (\log L)_- \circ L^i - 2c_i B_{i-1}.
\]
Therefore,
\begin{align*}
i! \frac{\p \psi_2}{\p b_i} = Q_i (\psi_2)={}&{}
 \sum_{j=0}^{i-1} \lambda^{i-j} \bigl(2m_j-2r_jw\Lambda^{-1}\bigr) (\psi_2)\\
&{}{-}\, 2 (\log L)_- (\psi_2)
-2c_i\sum_{j=0}^{i-1} \lambda^{i-1-j} \bigl(N_j-R_j w\Lambda^{-1}\bigr) (\psi_2),\nn
\end{align*}
and
\begin{align*}
i! \frac{\p \psi_1}{\p b_i} 
={}&{}2 \lambda^i(\log\lambda-c_i) \psi_1+ \sum_{j=0}^{i-1} \lambda^{i-j} \bigl(2m_j-2r_jw\Lambda^{-1}\bigr) (\psi_1)\\
&{}{-}\, 2 (\log L)_- (\psi_1)
-2c_i\sum_{j=0}^{i-1} \lambda^{i-1-j} \bigl(N_j-R_j w\Lambda^{-1}\bigr) (\psi_1).\nn
\end{align*}
Using these relations and using \eqref{Laxa1bi} and \eqref{Laxa2bi} with $i=0$, and again by a lengthy calculation, we find
\begin{align}
& {\rm e}^{-\Lambda^{-1}(\sigma)}\frac{\p \bigl({\rm e}^{\Lambda^{-1}(\sigma)}d\bigr)}{\p b_i}
= - 2\Lambda^{-1}(r_0)d + \p_x(d) = 0. \nn
\end{align}
The lemma is proved.
\end{proof}

It follows from Lemma~\ref{keylemma} and \eqref{psiphiform1} and \eqref{psiphiform2} that
\[
d(\lambda; x, {\bf b}, {\bf m};\e) = \lambda {\rm e}^{-\Lambda^{-1}(\sigma(x, {\bf b},{\bf m};\e))} \sum_{k\ge0} d_k \lambda^{-k},
\]
where $d_k$, $k\ge0$, are constants with $d_0=1$.
Therefore, for any choice $\psi_1$ of
the wave functions of type~A associated to $(v(x, {\bf b}, {\bf m};\e), w(x, {\bf b}, {\bf m};\e))$, there exists a choice, $\psi_2$, of
the wave functions of type~B associated to $(v(x, {\bf b}, {\bf m};\e), w(x, {\bf b}, {\bf m};\e))$, such that
\beq\label{dformula1002}
d(\lambda; x, {\bf b}, {\bf m};\e) = \lambda {\rm e}^{-\Lambda^{-1}(\sigma(x, {\bf b},{\bf m};\e))}.
\eeq
As a generalization of the terminology in~\cite{Y}, we say that $\psi_1$, $\psi_2$ form {\it a pair of wave functions}
associated to $(v({\bf b}+x{\bf 1}, {\bf m};\e), w({\bf b}+x{\bf 1}, {\bf m};\e))$ if~\eqref{dformula1002} holds.

We now define
\begin{align*}
&D(\lambda,\mu; x, {\bf b}, {\bf m};\e) \\
&\qquad{}:= \frac{\psi_1(\lambda; x, {\bf b}, {\bf m};\e) \psi_2(\mu; x-\e, {\bf b}, {\bf m};\e)
- \psi_1(\lambda; x-\e, {\bf b}, {\bf m};\e) \psi_2(\mu; x, {\bf b}, {\bf m};\e)}{\lambda-\mu}, \nn
\end{align*}
which coincides with the one introduced in~\cite{Y} when restricted to $b_2=b_3=\dots=0$.

We arrive at the following theorem.
\begin{Theorem} The following identity holds:
\begin{align}
& \e^n \sum_{i_1,\dots,i_n\ge0}
\prod_{j=1}^n \frac{(i_j+1)!}{\lambda_j^{i_j+2}} \frac{\p^n \log \tau({\bf b}+x{\bf 1},{\bf m}; \e)}{\p m_{i_1} \cdots \p m_{i_n}}
\nn\\
&\qquad{}= (-1)^{n-1}
\sum_{\text{$n$-cycles}}
\prod_{j=1}^n B\bigl(\lambda_{\sigma(j)},\lambda_{\sigma(j+1)}; x, {\bf b}, {\bf m};\e\bigr) - \frac{\delta_{n,2}}{(\lambda-\mu)^2},\label{Y926-ext-b}
\end{align}
where $n\ge2$ and
\begin{align*}
&B(\lambda,\mu;x, {\bf b},{\bf m};\e)\\
&\qquad{} := \frac{{\rm e}^{\sigma(x-\e, {\bf b},{\bf m};\e)}}{\mu} D(\lambda,\mu; x, {\bf b}, {\bf m};\e)
\frac
{{\rm e}^{\sum_{i\ge0} \frac2{i!} \frac{b_i}\e \mu^i (\log \mu-c_i)+\sum_{i\ge0} \frac{m_i}{(i+1)!\e} \mu^{i+1}}}
{{\rm e}^{\sum_{i\ge0} \frac2{i!} \frac{b_i}\e \lambda^{i} (\log \lambda-c_i)+\sum_{i\ge0} \frac{m_i}{(i+1)!\e} \lambda^{i+1}}}\frac{\mu^{\frac{x}\e}}{\lambda^{\frac{x}\e}}. \nn
\end{align*}
\end{Theorem}
\begin{proof}
Let
\[
\Psi_{\rm pair}(\lambda; x, {\bf b}, {\bf m};\e) :=
\begin{pmatrix} \psi_1(\lambda; x, {\bf b}, {\bf m};\e) & \psi_2(\lambda; x, {\bf b}, {\bf m};\e)
\\
\psi_1(\lambda; x-\e, {\bf b}, {\bf m};\e) & \psi_2(\lambda; x-\e, {\bf b}, {\bf m};\e)
\end{pmatrix} .
\]
By using the arguments the same as in~\cite{Y}, one can show that
\beq\label{RPsi}
R(\lambda; x, {\bf b}, {\bf m};\e) = \Psi_{\rm pair}(\lambda; x, {\bf b}, {\bf m};\e) \begin{pmatrix} 1 & 0 \\ 0 & 0\end{pmatrix}
\Psi_{\rm pair}(\lambda; x, {\bf b}, {\bf m};\e)^{-1}
\eeq
(see~\cite[the proof of Proposition~3]{Y}).
It follows from~\eqref{multiderivativesresolvent} and \eqref{RPsi} that for $n\ge2$,
\begin{align*}
&\e^n\sum_{i_1,\dots,i_n\ge0} \prod_{j=1}^n \frac{(i_j+1)!}{\lambda_j^{i_j+2}}
\frac{\p^n \log \tau({\bf b}+x{\bf 1},{\bf m})}{\p m_{i_1} \cdots \p m_{i_n}} \\
&\qquad{}= (-1)^{n-1}
\frac{{\rm e}^{n \sigma(x-\e,{\bf b}, {\bf m};\e)}}{\prod_{j=1}^n \lambda_j} \sum_{\text{$n$-cycles}}
\prod_{j=1}^n D\bigl(\lambda_{\sigma(j)},\lambda_{\sigma(j+1)}; x, {\bf b}, {\bf m};\e\bigr) - \frac{\delta_{n,2}}{(\lambda-\mu)^2}. \nn
\end{align*}
This yields identity~\eqref{Y926-ext-b}.
\end{proof}

By taking ${\bf m}={\bf 0}$, we immediately obtain the following corollary.
\begin{Corollary} We have
\begin{align}
& \e^n\sum_{i_1,\dots,i_n\ge0}
\prod_{j=1}^n \frac{(i_j+1)!}{\lambda_j^{i_j+2}} \frac{\p^n \log \tau({\bf b}+x{\bf 1},{\bf m};\e)}{\p m_{i_1} \cdots \p m_{i_n}}\bigg|_{{\bf m}={\bf 0}}
\nn\\
& \qquad{} = (-1)^{n-1}
\sum_{\text{$n$-cycles}}
\prod_{j=1}^n B\bigl(\lambda_{\sigma(j)},\lambda_{\sigma(j+1)}; x, {\bf b};\e\bigr) - \frac{\delta_{n,2}}{(\lambda-\mu)^2},\label{Y926-ext}
\end{align}
where $B(\lambda,\mu; x, {\bf b};\e):=B(\lambda,\mu; x, {\bf b},{\bf m}={\bf 0};\e)$.
\end{Corollary}

In terms of the functions $\phi_1$ and $\phi_2$, the pair condition~\eqref{dformula1002} reads
\begin{gather}
\phi_1(\lambda; x, {\bf b}, {\bf m};\e) \phi_2(\lambda; x-\e, {\bf b}, {\bf m};\e)\nn \\
\qquad{}- w(x, {\bf b}, {\bf m};\e)
\lambda^{-2} \phi_2(\lambda; x, {\bf b}, {\bf m};\e) \phi_1(\lambda; x-\e, {\bf b}, {\bf m};\e) \equiv 1,\label{pairconditionphi}
\end{gather}
and the function $B(\lambda,\mu;{\bf b},{\bf m};\e)$ reads
\begin{align}
&B(\lambda,\mu; x, {\bf b}, {\bf m};\e) \label{Bphiab} \\
&\quad{}:= \frac{\phi_1(\lambda; x, {\bf b}, {\bf m};\e) \phi_2(\mu; x-\e,{\bf b}, {\bf m};\e)
- \frac{w(x, {\bf b}, {\bf m};\e)}{\lambda\mu} \phi_1(\lambda; x-\e,{\bf b}, {\bf m};\e) \phi_2(\mu; x, {\bf b}, {\bf m};\e)}{\lambda-\mu}. \nn
\end{align}

\begin{Lemma}
The function $B(\lambda,\mu; x, {\bf b},{\bf m};\e)$ admits the following expansion:
\beq\label{eqn:hat-A-26-extwitht}
B(\lambda,\mu; x, {\bf b},{\bf m};\e) =
\frac1{\lambda-\mu} + \sum_{i,j\ge0} \frac{A_{i,j}(x, {\bf b},{\bf m};\e)}{\lambda^{j+1} \mu^{i+1}} .
\eeq
In particular,
\beq\label{eqn:hat-A-26-ext}
B(\lambda,\mu;x, {\bf b};\e) =
\frac1{\lambda-\mu} + \sum_{i,j\ge0} \frac{A_{i,j}(x, {\bf b};\e)}{\lambda^{j+1} \mu^{i+1}} ,
\eeq
where
$A_{i,j}(x, {\bf b};\e)=A_{i,j}(x, {\bf b},{\bf m}={\bf 0};\e)$.
\end{Lemma}
\begin{proof}
Recall the following identity:
\begin{align*}
\phi_2(\mu) = \phi_2(\lambda)+\phi_2'(\lambda)(\mu-\lambda)
+ (\mu-\lambda)^2 \p_\lambda \left(\frac{\phi_2(\lambda)-\phi_2(\mu)}{\lambda-\mu}\right),
\end{align*}
where we omit the arguments $x$, ${\bf b}$, ${\bf m}$ from $\phi_2(\lambda; x, {\bf b}, {\bf m};\e)$.
Similarly,
\[
\frac{\phi_2(\mu)}{\mu} = \frac{\phi_2(\lambda)}{\lambda}+\p_\lambda\biggl(\frac{\phi_2(\lambda)}{\lambda}\biggr) (\mu-\lambda)
+ (\mu-\lambda)^2 \p_\lambda \biggl(\frac{\phi_2(\lambda)/\lambda-\phi_2(\mu)/\mu}{\lambda-\mu}\biggr).
\]
Substituting these identities in~\eqref{Bphiab} and using \eqref{pairconditionphi}, \eqref{phiexpansion1} and \eqref{phiexpansion2}, we find the validity of the expansion~\eqref{eqn:hat-A-26-extwitht}.
\end{proof}

\begin{proof}[Proof of Theorem~\ref{thm:toda-kp-ext}]
By using \eqref{Y926-ext} and \eqref{eqn:hat-A-26-ext} and Corollary~\ref{conversthmA}.
\end{proof}

Like in the previous section, if we write
\[
\tau({\bf b}+x{\bf 1},{\bf m};\e)=\tau_0({\bf b}+x{\bf 1};\e) \tau_1({\bf b}+x{\bf 1}, {\bf m};\e),
\qquad \tau_1({\bf b}+x{\bf 1}, {\bf 0};\e) \equiv 1,
\]
then by Theorem~\ref{thm:toda-kp-ext} we know that $\tau_1({\bf b}+x{\bf 1}, {\bf m};\e)$ is a KP tau-function in the big cell.
The factor $\tau_0({\bf b}+x{\bf 1};\e)$ can be determined from the definition of the tau-structure~\cite{CDZ}.
Let $\tilde{\tau}_1(x, {\bf b}, {\bf m};\e)$ be the~KP tau-function
associated to the point in the infinite Grassmannian with affine coordinates
$A_{i,j}(x, {\bf b};\e)$, with $A_{i,j}(x, {\bf b};\e)$ given in~\eqref{eqn:hat-A-26-ext}. Then there exists
$a_i(x, {\bf b};\e)$, $i\ge0$, such~that
\[
\tau_1({\bf b}+x{\bf 1}, {\bf m};\e) = {\rm e}^{\sum_{i\ge0} a_i(x, {\bf b};\e) m_i} \tilde{\tau}_1(x, {\bf b}, {\bf m};\e).
\]
By further applying Proposition~\ref{VVprimelinear}, we can find a formula for the affine coordinates for
$\tau_1({\bf b}+x{\bf 1}, {\bf m};\e)$, whose
dependence on $b_0, x$ is through $b_0+x$.

\section[Application to Gromov--Witten invariants of P\string^{}1]{Application to Gromov--Witten invariants of~$\boldsymbol{\mathbb{P}^1}$} \label{sectionP1}
In this section, we give an application of the above results regarding the topological solution to the ETH.

Recall that the {\it free energy $\mathcal{F}^{\bP^1}=\mathcal{F}^{\bP^1}({\bf b}, {\bf m};\e,q)$ of GW invariants of~$\bP^1$} is defined by
\begin{gather*}
\mathcal{F}^{\bP^1} = \sum_{k\geq 1} \frac1{k!}
\sum_{1\leq \alpha_1,\dots,\alpha_k\leq 2 \atop i_1,\dots,i_k\geq 0} \prod_{j=1}^k t^{\alpha_j,i_j} \sum_{g\ge0} \e^{2g-2} \sum_{d\ge0} q^d
\langle \tau_{i_1} (\alpha_1) \cdots \tau_{i_k} (\alpha_k) \rangle^{\bP^1}_{g,d}
\end{gather*}
where 
$b_i=t^{1,i}$, $m_i=t^{2,i}$, $i\ge0$,
and $\langle \tau_{i_1} (\alpha_1) \cdots \tau_{i_k} (\alpha_k) \rangle^{\bP^1}_{g,d}$ are
genus~$g$ and degree~$d$ GW invariants of~$\bP^1$ 
\cite{DYZ, DZ-toda, OP2, OP1}.
When $\alpha_1=\dots=\alpha_k=2$ the GW invariants \smash{$\langle \tau_{i_1} (\alpha_1) \cdots \tau_{i_k} (\alpha_k) \rangle^{\bP^1}_{g,d}$}
are~called in the {\it stationary sector}~\cite{OP2, OP1}, and are denoted as
\smash{$\langle \prod_{j=1}^n \tau_{i_j} (H) \rangle^{\bP^1}_{g,d}$} in some literatures.
The exponential
\[
\exp \mathcal{F}^{\bP^1}({\bf b}, {\bf m};\e,q) =: Z^{\bP^1}({\bf b}, {\bf m};\e,q)
\]
is called the {\it partition function of GW invariants of~$\bP^1$}. The restrictions
\[
\mathcal{F}^{\bP^1}({\bf b}={\bf 0}, {\bf m};\e,q)
=: \mathcal{F}^{\bP^1}({\bf m};\e,q), \qquad Z^{\bP^1}({\bf b}={\bf 0}, {\bf m};\e,q) =:
Z^{\bP^1}({\bf m};\e,q)
\]
are called the free energy and respectively the partition function of stationary GW invariants of~$\bP^1$,
which are of particular interest due to their closed formulas~\cite{DYZ} and their relations to
the representations of the symmetric group~\cite{OP1}.

It is proved by Dubrovin and Zhang~\cite{DZ-toda} that assuming Virasoro constraints
the partition function
$Z^{\bP^1}({\bf b}+x{\bf 1},{\bf m};\e,q=1)=:Z^{\bP^1}({\bf b}+x{\bf 1},{\bf m};\e)$ of GW invariants of~$\bP^1$ is a particular tau-function
for the ETH.
\big(Note that by degree-dimension counting the Novikov variable~$q$ can be recovered via a rescaling \smash{$\epsilon\to q^{-1/2}\e$}, \smash{$b_i\to q^{(i-1)/2} b_i$}, and \smash{$m_i\to q^{i/2} m_i$}.\big)
Then, as given in Introduction, the result of Dubrovin--Zhang together with Theorem~\ref{thm:toda-kp-ext}
implies Corollary~\ref{ThmMainP1}.
As before, we can write the partition function $Z^{\bP^1}({\bf b}+x{\bf 1},{\bf m};\e)$ as follows:
\[
Z^{\bP^1}({\bf b}+x{\bf 1}, {\bf m};\e) = Z^{\bP^1}_{\rm corr}({\bf b}+x{\bf 1};\e) \cdot Z^{\bP^1}_{\rm 1}({\bf b}+x{\bf 1}, {\bf m};\e),
\]
where the factor $Z^{\bP^1}_{\rm corr}({\bf b}+x{\bf 1};\e)$ can be determined by using the definition of tau-function for the ETH~\cite{DZ-toda},
and the factor $Z^{\bP^1}_{\rm 1}({\bf b}+x{\bf 1}, {\bf m};\e)$, which
satisfies $Z^{\bP^1}_{\rm 1}({\bf b}+x{\bf 1}, {\bf 0};\e)\equiv 1$,
is a particular KP tau-function in the big cell.

 Denote $Z^{\bP^1}_{\rm 1}({\bf b}+x{\bf 1}, {\bf m};\e,q)$ the normalized factor with $q$ recovered,
 and by $A^{\bP^1}_{i,j}({\bf b}+x{\bf 1};\e,q)$ the affine
 coordinates of \smash{$Z^{\bP^1}_{\rm 1}({\bf b}+x{\bf 1}, {\bf m};\e,q)$}.
The restrictions of these affine coordinates to ${\bf b}={\bf 0}$, $x=0$
are of particular interest, and will be denoted by \smash{$A^{\bP^1}_{i,j}(\e,q)$}.
In the next two subsections, we will derive two formulas for \smash{$A^{\bP^1}_{i,j}(q):=A^{\bP^1}_{i,j}(1,q)$}.

\subsection[Explicit generating series for A\string^{}\{P\string^1\}\_\{i,j\}(q)]{Explicit generating series for $\boldsymbol{A^{\bP^1}_{i,j}(q)}$}
It was proved in~\cite[Theorems 1--5 and equation~(93)]{DYZ} (see also~\cite{DY2, Marchal, Y}) that, for each $n\ge2$,
\begin{align*}
& \sum_{i_1,\dots,i_n\geq0} \prod_{j=1}^n \frac{(i_j+ 1)!}{\xi^{i_j+2}_j} \,
\sum_{g\geq 0}\sum_{d\geq 0} q^d
\Biggl\langle \prod_{j=1}^n \tau_{i_j} (2) \Biggr\rangle^{\bP^1}_{g,d} \\
&\qquad{}= (-1)^{n-1}
\sum_{\text{$n$-cycles}}
\prod_{j=1}^n B\bigl(\xi_{\sigma(j)}, \xi_{\sigma(j+1)};q\bigr) - \frac{\delta_{n,2}}{(\xi_1-\xi_2)^2}, \nn
\end{align*}
where $B(\lambda,\mu;q)$ is the asymptotic expansion of the following analytic function
\be\label{anaB}
-\frac1{\mu-\lambda}\, {}_2 F_3\biggl(\frac{\lambda-\mu}2, \frac{\lambda-\mu+1}2; \, \frac12 - \mu \,, \, \frac12 + \lambda \,, \lambda-\mu+1; \, -4 q\biggr)
\ee
for $\lambda,\mu\notin\mathbb{Z}+\frac12$ as $\lambda,\mu\to\infty$,
and is explicitly given by
\begin{align}
B(\lambda, \mu;q) ={}& \frac{1}{\lambda - \mu} \label{Dlargeabexpansion} \\
&{}{-}
\sum_{i,j\geq 0} \frac{(-1)^{j+1}}{\lambda^{j+1} \mu^{i+1} }
\sum_{k\geq 1} \frac{q^{k}}{k!}\sum_{1\leq i_1,j_1 \leq k}\!
(-1)^{i_1+j_1} \frac{(i_1 + j_1 - 2k)_{k-1} \bigl(i_1 -\frac{1}{2}\bigr)^{i}
\bigl(j_1 - \frac{1}{2}\bigr)^{j}}{(i_1-1)!(j_1-1)!(k-i_1)!(k-j_1)!}.\nn
\end{align}

If we define \smash{$A^{\bP^1,+}_{i,j}(q)$} via
\be \label{eqn:hat-A}
\sum_{i,j \geq 0} A^{\bP^1,+}_{i,j}(q) \lambda^{-j-1} \mu^{-i-1} := B(\lambda,\mu; q) - \frac{1}{\lambda - \mu}
\ee
and let $Z^{\bP^1,+}({\bf m},q)$ be the KP tau-function associated to the point in $Gr_0$
with the affine coordinates \smash{$A^{\bP^1,+}_{i,j}(q)$} (i.e., using formula~\eqref{fromaffinetotau}),
then 
we know that \smash{$Z^{\bP^1}_1({\bf 0},{\bf m};1, q)=Z^{\bP^1}({\bf m};1,q)$} and~\smash{$Z^{\bP^1,+}({\bf m},q)$} can only differ by multiplication by the exponential of
a particular linear function of~${\bf m}$.
\begin{Proposition} \label{positive} We have
\beq\label{toshowpositive}
\log Z^{\bP^1,+}({\bf m},q) = \sum_{n\geq 1} \frac1{n!}
\sum_{i_1,\dots,i_n\geq 0} \prod_{j=1}^n m_{i_j} \sum_{g\ge0}\sum_{d>0} q^d
\Biggl\langle \prod_{j=1}^n \tau_{i_j} (2) \Biggr\rangle^{\bP^1}_{g,d} .
\eeq
\end{Proposition}
\begin{proof}
We already know that $\log Z^{\bP^1}({\bf m};1,q)- \log Z^{\bP^1,+}({\bf m},q)$ is a linear function of~${\bf m}$.
Below, we will often omit the arguments ``$1$, $q$'', ``$q$'' for simplifying the notations.
According to, e.g.,~\cite[p.~529]{OP1}, \smash{$\langle \prod_{j=1}^n \tau_{i_j} (2) \rangle^{\bP^1}_{g,0}$}
vanish whenever $n\ge2$, thus we get the validity of the nonlinear part of~\eqref{toshowpositive}.
By using~\cite[equations~(31), (32) and (34)]{DYZ} and
the asymptotic expansion of the digamma function
\[
\psi\biggl(\xi+\frac12\biggr) \sim \log \xi + \sum_{m\ge2} (-1)^{m-1} \bigl(2^{1-m}-1\bigr) \frac{B_m}m \xi^{-m}
\]
as $\xi$ being large with $\operatorname{arg} \xi <\pi-\epsilon$, we calculate out the linear function
\begin{gather}\label{onepoint1004}
\log Z^{\bP^1}({\bf m}) - \log Z^{\bP^1,+}({\bf m}) = - \sum_{k \geq 1} \bigl(1-2^{-k}\bigr) \zeta(-k) \frac{m_{k-1}}{(k-1)!},
\end{gather}
with $\zeta(-k)$ given by
\[
\zeta(-k) = (-1)^k \frac{B_{k+1}}{k+1}, \qquad k\ge1.
\]
Here $B_j$ denotes the $j$-th Bernoulli number.
The proposition is then
proved by identifying this linear function with the degree
zero part of stationary $\bP^1$ free energy (see~\cite[p.~163]{DYZ} or~\cite[equation~(0.26)]{OP1}).
\end{proof}

Proposition~\ref{positive}, formula~\eqref{onepoint1004} and Proposition~\ref{VVprimelinear} imply the following theorem.
\begin{Theorem}\label{affine1}
The generating series of the affine coordinates \smash{$A^{\bP^1}_{i,j}(q)$}
is given by
\begin{gather*}
\sum_{i,j \geq 0} A^{\bP^1}_{i,j}(q) \xi^{-j-1} \eta^{-i-1}
=B(\xi,\eta;q) \, \frac{{\rm e}^{-\bigl(\log \frac{\Gamma(\xi+\frac{1}{2})}{\sqrt{2\pi}} -\xi\log \xi + \xi\bigr)}}
{{\rm e}^{-\bigl(\log \frac{\Gamma(\eta+\frac{1}{2})}{\sqrt{2\pi}} -\eta\log \eta + \eta\bigr)}} - \frac1{\xi-\eta},
\end{gather*}
where $B(\xi,\eta;q)$ is given by~\eqref{anaB} $($cf.\ $\eqref{Dlargeabexpansion})$, and $\Gamma\bigl(\xi+\frac{1}{2}\bigr)$ is understood as its asymptotic expansion for $\xi$ large with $\operatorname{arg} \xi <\pi-\epsilon$.
\end{Theorem}

\begin{Remark}
We consider the computations of $\log Z^{\bP^1}_0({\bf b})$ and $A^{\bP^1}_{i,j}({\bf b})$
 by using Virasoro constraints~\cite{DZ-toda} and the ETH in a future publication.
\end{Remark}

\subsection{Affine coordinates from the Okounkov--Pandharipande formula}
We proceed with recalling the Okounkov--Pandharipande formula~\cite{OP1} for $Z^{\bP^1}({\bf m})$, which is
derived from the GW/H (Gromov--Witten/Hurwitz)
correspondence established by Okounkov and Pandharipande~\cite{OP1}.

To consider the GW invariants of~$\bP^1$,
one has to consider all the moduli spaces $\Mbar_{g,n}\bigl(\bP^1, d\bigr)$ of stable maps,
where $g$ is the genus of the domain curve,
$n$ is the number of marked point on the domain curve,
$d$ is the degree of the map.
For the purpose of GW/H correspondence,
it is necessary to consider the $n=0$ case.
The expected dimension of $\Mbar_{g,0}\bigl(\bP^1, d\bigr)$ is
$2g-2+2d$.
When
\[
2g-2+2d = 0,
\]
we have a contribution of $q^d$ to the free energy.
This happens only when $(g,d) = (0,1)$ and ${(g,d)=(1,0)}$,
The first case corresponds to $\Mbar_{0,0}\bigl(\bP^1,1\bigr)$,
and the second case is impossible.

With the contribution of zero point correlators,
the free energy of the stationary GW invariants of $\bP^1$ is
then
\[
\tilde{\F}^{\bP^1}({\bf m}) = q + \F^{\bP^1}({\bf m}).
\]
Let $\tilde{Z}^{\bP^1}({\bf m}) = \exp \tilde{F}^{\bP^1}({\bf m})$ be the corrected partition function.
Then
\be\label{deftildeZ}
\tilde{Z}^{\bP^1}({\bf m}) = {\rm e}^q \cdot Z^{\bP^1}({\bf m}).
\ee
As a corollary to Corollary~\ref{ThmMainP1}, we get the following.

\begin{Corollary}
The corrected partition function $\tilde{Z}^{\bP^1}({\bf m})$ is a KP tau-function.
\end{Corollary}

The corrected partition function $\tilde{Z}^{\bP^1}$ has the following
expansion:
\[
\tilde{Z}^{\bP^1}({\bf m}) = \sum_{d\geq 0} q^d\sum_{n\ge0} \sum_{i_1,\dots, i_n \geq 0} \frac{m_{i_1} \cdots m_{i_n}}{n!}
\Biggl\langle\prod_{j=1}^n \tau_{i_j}(2)\Biggr\rangle^{\bullet \bP^1}_d,
\]
where \smash{$\langle\prod_{j=1}^n \tau_{i_j}(2)\rangle^{\bullet \bP^1}_d$} denotes the
{\it not-necessarily connected} GW invariants of~$\bP^1$ of degree~$d$ in the stationary sector.
By using the GW/H correspondence
Okounkov and Pandharipande obtained~\cite{OP1} the following formula:
\be\label{OPformula1004}
\Biggl\langle\prod_{j=1}^n \tau_{i_j}(2)\Biggr\rangle^{\bullet \bP^1}_d
= \sum_{\lambda \in \mathcal{P}_d} \biggl(\frac{\dim V^\lambda}{d!}\biggr)^2
\prod_{j=1}^n \frac{\bp_{i_j+1}(\lambda)}{(i_j+1)!},
\ee
where $V^\lambda$ is the irreducible representation of $S_d$ indexed by $\lambda$,
and $\bp_k(\lambda)$ is defined by
\[
\bp_k(\lambda)
= \sum_{i\geq 1} \biggl[\biggl(\lambda_i-i+\half\biggr)^k - \biggl(-i+\half\biggr)^k\biggr] + \bigl(1-2^{-k}\bigr) \zeta(-k),
\]
where $k\ge1$.
Dubrovin~\cite{Du16} gave a different proof of formula~\eqref{OPformula1004}
using symplectic field theory.
(We would like to mention that $\bp_k(\lambda)$ are generators of shifted symmetric functions
that play a~key role in the Bloch--Okounkov theorem~\cite{BO}, and are explained in detail in,
e.g.,~\cite{BO, Zagier}.)

We now have
\begin{align*}
\tilde{Z}^{\bP^1}({\bf m})
& = \sum_{d\geq 0} q^d \sum_{n\ge0} \sum_{i_1,\dots, i_n \geq 0} \frac{m_{i_1} \cdots m_{i_n}}{n!}
\sum_{\lambda \in \mathcal{P}_d} \biggl(\frac{\dim V^\lambda}{d!}\biggr)^2
\prod_{j=1}^n \frac{\bp_{i_j+1}(\lambda)}{(i_j+1)!} \nn \\
& = \sum_{d\geq 0} \frac{q^d}{d!}
\sum_{\lambda \in \mathcal{P}_d} \frac{\bigl(\dim V^\lambda\bigr)^2}{d!}
\exp \sum_{i\ge0} \frac{m_i}{(i+1)!} \bp_{i+1}(\lambda). 
\end{align*}
When all $m_i$'s are taken to be $0$,
we get the vacuum expectation value
\[
\tilde{Z}^{\bP^1}({\bf m}={\bf 0}) = \sum_{d\geq 0} \frac{q^d}{d!}
\sum_{\lambda \in \mathcal{P}_d} \frac{\bigl(\dim V^\lambda\bigr)^2}{d!} = {\rm e}^q.
\]

By the definition~\eqref{deftildeZ}, we have
\begin{align*}
Z^{\bP^1}({\bf m}) & =
 {\rm e}^{-q} \tilde{Z}^{\bP^1}({\bf m}) \\
& = \sum_{d\geq 0} \frac{q^d}{d!}{\rm e}^{-q}
\sum_{\lambda \in \mathcal{P}_d} \frac{\bigl(\dim V^\lambda\bigr)^2}{d!}
\exp \sum_{i\ge0} \frac{m_i}{(i+1)!} \bp_{i+1}(\lambda).
\end{align*}
Now note that
\smash{$\bigl\{\bigl\{\frac{q^d}{d!}{\rm e}^{-q} \frac{(\dim V^\lambda)^2}{d!} \bigr\}_{\lambda \in \mathcal{P}_d}\bigr\}_{d \geq 0}$}
is the Poissonized Plancherel measure on the set of partitions and
\smash{$\exp \sum_{i\ge0} \frac{m_i}{(i+1)!} \bp_{i+1}$}
is a random variable on the set of partitions,
so the partition function $ Z^{\bP^1}({\bf m})$ is just the expectation value,
\[
 Z^{\bP^1}({\bf m}) = \Corr{\exp \sum_{i\ge0} \frac{m_i}{(i+1)!} \bp_{i+1}(\lambda)}_{\rm Plancherel}.
\]

Now we notice that \smash{$p_k = \frac{m_{k-1}}{(k-1)!}$}, $k\ge1$,
so we get
\[
 Z^{\bP^1}({\bf m}) = \sum_{d\geq 0} \frac{q^d}{d!}{\rm e}^{-q}
\sum_{\lambda \in \mathcal{P}_d} \frac{\bigl(\dim V^\lambda\bigr)^2}{d!}
\exp \sum_{k\ge1} \frac{p_k}{k} \bp_{k}(\lambda).
\]
Recall
\[
\exp \sum_{k\ge1} \frac1k p_k p_k' = \sum_{\mu\in \mathcal{P}} s_\mu({\bf p})s_\mu({\bf p}'),
\]
where $s_\mu$ is the Schur polynomial (cf.\ \eqref{defSchur1}) which has the expression
\be\label{Schurdef2}
s_\mu({\bf p}) = \sum_{\nu\in\mathcal{P}} \frac{\chi^\mu_\nu}{z_\nu} p_\nu, \qquad \mu\in\mathcal{P}
\ee
(see Macdonald~\cite{Mac} for the details).
Here,
$\chi^\mu_\nu$ denotes the character table, $z_\nu\!=\!\prod_{i=1}^\infty i^{m_i(\nu)} m_i(\nu)!$, and
$p_\nu=p_{\nu_1}\cdots p_{\nu_{\ell(\nu)}}$.
It follows that 
\begin{align*}
 Z^{\bP^1}({\bf m}) = \sum_{d\geq 0} \frac{q^d}{d!}{\rm e}^{-q}
\sum_{\lambda\in \mathcal{P}_d} \frac{\bigl(\dim V^\lambda\bigr)^2}{d!}
 \sum_{\mu\in\mathcal{P}} s_\mu({\bf p}) \bs_{\mu}(\lambda)
 = \sum_{\mu\in\mathcal{P}} \corr{\bs_\mu}_{\rm Plancherel} \cdot s_\mu({\bf p}).
\end{align*}
Together with Corollary~\ref{ThmMainP1}, we have the following.

\begin{Theorem}\label{affinemeasure}
The affine coordinates $A_{i,j}^{\bP^1}$ of $Z^{\bP^1}({\bf m})$ are given by the Plancherel expectation value
of $\bs_{(i+1, 1^j)}$:
\[
A_{i,j}^{\bP^1} = (-1)^j \bigl\langle\bs_{(i+1,1^j)}\bigr\rangle_{\rm Plancherel},
\]
or explicitly,
\[
A_{i,j}^{\bP^1} = (-1)^j \sum_{d\geq 0} \frac{q^d}{d!}{\rm e}^{-q}
\sum_{\lambda \in \mathcal{P}_d} \frac{\bigl(\dim V^\lambda\bigr)^2}{d!} \bs_{(i+1,1^j)}(\lambda).
\]
\end{Theorem}

Combining Theorems~\ref{affine1} and~\ref{affinemeasure}, we have the following.

\begin{Corollary}\label{coridrepele}
The following identity holds:
\begin{align}
&\frac{1}{\xi-\eta} + \sum_{i,j\ge0} (-1)^j \sum_{d\geq 0} \frac{q^d}{d!}{\rm e}^{-q}
\sum_{\lambda \in \mathcal{P}_d} \frac{\bigl(\dim V^\lambda\bigr)^2}{d!} \frac{\bs_{(i+1,1^j)}(\lambda)}{\xi^{j+1} \eta^{i+1}} \nonumber\\
&\qquad{}= B(\xi, \eta;q) \,
\frac
{{\rm e}^{-\bigl(\log \frac{\Gamma(\xi+\frac{1}{2})}{\sqrt{2\pi}} -\xi\log \xi + \xi\bigr)}}
{{\rm e}^{-\bigl(\log \frac{\Gamma(\eta+\frac{1}{2})}{\sqrt{2\pi}} -\eta\log \eta + \eta\bigr)}},\label{id1006}
\end{align}
where $\Gamma\bigl(\xi+\frac{1}{2}\bigr)$ is understood as its asymptotic expansion for $\xi$ large with $\operatorname{arg} \xi <\pi-\epsilon$, and~$B(\xi,\eta;q)$ is given by~\eqref{Dlargeabexpansion} $($cf.\ $\eqref{anaB})$.
\end{Corollary}
For example,
\begin{gather*}
A^{\bP^1}_{0,0} = \sum_{d\geq 0} \frac{q^d}{d!}{\rm e}^{-q}
\sum_{\lambda \in \mathcal{P}_d} \frac{\bigl(\dim V^\lambda\bigr)^2}{d!} \bs_{(1)}(\lambda)\\ \hphantom{A^{\bP^1}_{0,0} }{}
 = {\rm e}^{-q}\biggl[-\frac{1}{24} + \frac{23}{24}q + \biggl( \biggl(\frac{1}{2!}\biggr)^2 \frac{47}{24}
+ \biggl(\frac{1}{2!}\biggr)^2 \frac{47}{24}\biggr) q^2 \\ \hphantom{A^{\bP^1}_{0,0} = {\rm e}^{-q}\biggl[}{}
 + \biggl( \biggl(\frac{1}{3!}\biggr)^2 \frac{71}{24}
+ \biggl(\frac{2}{3!}\biggr)^2\frac{71}{24}
+ \biggl(\frac{1}{3!}\biggr)^2 \frac{71}{24}\biggr) q^3 + \cdots\biggr] \\\hphantom{A^{\bP^1}_{0,0} }{}
 = -\frac{1}{24}+q,
\\
A^{\bP^1}_{1,0} = \sum_{d\geq 0} \frac{q^d}{d!}{\rm e}^{-q}
\sum_{\lambda \in \mathcal{P}_d} \frac{\bigl(\dim V^\lambda\bigr)^2}{d!} \bs_{(2)}(\lambda)\\ \hphantom{A^{\bP^1}_{1,0}}{}
 = \sum_{d\geq 0} \frac{q^d}{d!}{\rm e}^{-q}
\sum_{\lambda \in \mathcal{P}_d} \frac{\bigl(\dim V^\lambda\bigr)^2}{d!} \biggl(\frac{1}{2}\bp_{(2)}(\lambda)
+ \frac{1}{2} \bp_1(\lambda)^2\biggr) \\ \hphantom{A^{\bP^1}_{1,0}}{}
 = \frac{1}{2} {\rm e}^{-q} \biggl[0+ 0 \cdot q + \biggl( \biggl(\frac{1}{2!}\biggr)^2 \cdot 2
+ \biggl(\frac{1}{2!}\biggr)^2 \cdot (-2) \biggr) q^2 \\ \hphantom{A^{\bP^1}_{1,0}=\frac{1}{2} {\rm e}^{-q} \biggl[}{}
 + \bigg( \biggl(\frac{1}{3!}\biggr)^2 \cdot 6
+ \biggl(\frac{2}{3!}\biggr)^2\cdot 0
+ \biggl(\frac{1}{3!}\biggr)^2 \cdot (-6) \biggr) q^3 + \cdots \\ \hphantom{A^{\bP^1}_{1,0}=\frac{1}{2} {\rm e}^{-q} \biggl[}{}
 + \biggl(-\frac{1}{24}\biggr)^2 + \biggl(\frac{23}{24}\biggr)^2q
+ \biggl( \biggl(\frac{1}{2!}\biggr)^2 \biggl(\frac{47}{24}\biggr)^2
+ \biggl(\frac{1}{2!}\biggr)^2 \biggl(\frac{47}{24}\biggr)^2 \biggr) q^2 \\ \hphantom{A^{\bP^1}_{1,0}=\frac{1}{2} {\rm e}^{-q} \biggl[}{}
 + \biggl( \biggl(\frac{1}{3!}\biggr)^2 \biggl(\frac{71}{24}\biggr)^2
+ \biggl(\frac{2}{3!}\biggr)^2 \biggl(\frac{71}{24}\biggr)^2
+ \biggl(\frac{1}{3!}\biggr)^2 \biggl(\frac{71}{24}\biggr)^2\biggr) q^3 + \cdots\biggr] \\ \hphantom{A^{\bP^1}_{1,0}}{}
 = \frac{1}{1152}+\frac{11q}{24} + \frac{q^2}{2}.
\end{gather*}
These match with the right-hand side of~\eqref{id1006}.

Inspired by Han's conjecture~\cite{Han}, it was proved by Stanley~\cite{S} that
\[
\corr{\bp_\mu}_d:= \sum_{\lambda \in \mathcal{P}_d} \frac{\bigl(\dim V^\lambda\bigr)^2}{d!}
\bp_{\mu}(\lambda)
\]
is a polynomial in $d$ of degree $|\mu|$, where \smash{$\bp_{\mu}(\lambda):=\prod_{j=1}^{\ell(\mu)}\bp_{\mu_j}(\lambda)$}.
Note that Han's conjecture is on the polynomiality of $\corr{\bp_k}_d$ (confirmed also in~\cite{Pan}). It follows that
\[
\corr{\bs_\mu}_d:= \sum_{\lambda \in \mathcal{P}_d} \frac{\bigl(\dim V^\lambda\bigr)^2}{d!}
\bs_{\mu}(\lambda)
\]
is a polynomial in $d$ of degree $|\mu|$, where $\bs_{\mu}(\lambda):=s_\mu (\bp(\lambda))$ with $s_\mu$
denoting Schur functions as in~\eqref{Schurdef2} or~\eqref{defSchur1}.
For example,
\begin{gather*}
\bigl\langle\bp_{(1)}\bigr\rangle_d = d- \frac{1}{24}, \\
\bigl\langle\bp_{(2)}\bigr\rangle_d = 0, \\
\bigl\langle\bp_{(1^2)}\bigr\rangle_d = \biggl(d-\frac{1}{24}\biggr)^2, \\
\bigl\langle\bp_{(3)}\bigr\rangle_d = \frac{3d^2}{2}
- \frac{5d}{4} + \frac{7}{960},\\
\bigl\langle\bp_{(21)}\bigr\rangle_d = 0, \\
\bigl\langle\bp_{(1^3)}\bigr\rangle_d = d^3-\frac{1}{8}d^2
+ \frac{1}{192}d -\frac{1}{13824},
\end{gather*}
and from these explicit expressions we get
\begin{gather*}
\bigl\langle\bs_{(1)}\bigr\rangle_d = d- \frac{1}{24}, \\
\bigl\langle\bs_{(2)}\bigr\rangle_d = \bigl\langle\bs_{(1^2)}\bigr\rangle_d = \biggl(d-\frac{1}{24}\biggr)^2, \\
\bigl\langle\bs_{(3)}\bigr\rangle_d = \bigl\langle\bs_{(1^3)}\bigr\rangle
=\frac{d^3}{6} +\frac{23}{48}d^2-\frac{479}{1152}d
+\frac{1003}{414720}, \\
\bigl\langle\bs_{(21)}\bigr\rangle_d =
\frac{d^3}{3} - \frac{13}{24}d^2
+\frac{241}{576}d
-\frac{509}{207360} .
\end{gather*}
When $\mu$ are hook partitions, Corollary~\ref{coridrepele} not only confirms polynomiality
of $\corr{\bs_\mu}_d$ (so of the original Han's conjecture as well) but also leads to
elementary formulas for them. The explicit expressions, as well as general ones using results in~\cite{DY, DY2, DYZ, Marchal}, will be given
in a subsequent publication.

Note that by~\eqref{onepoint1004},
we have
\[
Z^{\bP^1}({\bf m}) = Z^{+, \bP^1}({\bf m}) \cdot \exp \sum_{k \geq 1} \bigl(1-2^{-k}\bigr) \zeta(-k) \frac{m_{k-1}}{(k-1)!},
\]
so we can also define
\[
\hat{\bp}_k(\lambda)
= \sum_{i\geq 1}\biggl[\biggl(\lambda_i-i+\half\biggr)^k - \biggl(-i+\half\biggr)^k\biggr].
\]
Then we repeat everything we have done in this subsection %
by replacing
every $\bp_k$ with $\hat{\bp}_k$ to~get
\ben
Z^{\bP^1,+}({\bf m})
& = & \sum_{d\geq 0} \frac{q^d}{d!}{\rm e}^{-q}
\sum_{\lambda\in \mathcal{P}_d} \frac{\bigl(\dim V^\lambda\bigr)^2}{d!}
\exp \sum_{i\ge0} \frac{m_i}{(i+1)!} \hat{\bp}_{i+1}(\lambda),
\een
and
\[
Z^{\bP^1,+}({\bf m}) = \Corr{\exp \sum_{i\ge0} \frac{m_i}{(i+1)!}
 \hat{\bp}_{i+1}(\lambda)}_{\rm Plancherel}.
\]
And if we set
\[
\hat{\bs}_\mu = \sum_{\nu\in\mathcal{P}} \frac{\chi^\mu_\nu}{z_\nu} \hat{\bp}_\nu,
\]
then we have
\begin{align*}
Z^{\bP^1, +}({\bf m}) = \sum_{d\geq 0} \frac{q^d}{d!}{\rm e}^{-q}
\sum_{\lambda\in \mathcal{P}_d} \frac{\bigl(\dim V^\lambda\bigr)^2}{d!}
 \sum_{\mu\in\mathcal{P}} s_\mu({\bf p}) \hat{\bs}_{\mu}(\lambda)
 = \sum_{\mu\in\mathcal{P}} \corr{\hat{\bs}_\mu}_{\rm Plancherel} \cdot s_\mu({\bf p}).
\end{align*}

\begin{Theorem}\label{affinemeasure2}
The affine coordinates \smash{$A_{i,j}^{\bP^1, +}$} of $Z^{\bP^1,+}({\bf m})$ are given by the Plancherel expectation value
of $\hat{\bs}_{(i+1, 1^j)}$,
\[
A_{i,j}^{\bP^1, +} = (-1)^j \bigl\langle\hat{\bs}_{(i+1,1^j)}\bigr\rangle_{\rm Plancherel},
\]
or explicitly,
\[
A_{i,j}^{\bP^1, +} = (-1)^j \sum_{d\geq 0} \frac{q^d}{d!}{\rm e}^{-q}
\sum_{\lambda\in \mathcal{P}_d} \frac{\bigl(\dim V^\lambda\bigr)^2}{d!} \hat{\bs}_{(i+1,1^j)}(\lambda).
\]
\end{Theorem}

Combining \eqref{eqn:hat-A} and Theorem~\ref{affinemeasure2}, we have the following.
\begin{Corollary} \label{cor:Plancherel}
The following identity holds:
\begin{align}
&\sum_{i,j} (-1)^j \sum_{d\geq 0} \frac{q^d}{d!}{\rm e}^{-q}
\sum_{\lambda\in \mathcal{P}_d} \frac{\bigl(\dim V^\lambda\bigr)^2}{d!} \frac{\hat{\bs}_{(i+1,1^j)}(\lambda)}{\xi^{j+1} \eta^{i+1}} = B(\xi,\eta;q), \label{id21004}
\end{align}
where $B(\xi,\eta;q)$ is given by~\eqref{Dlargeabexpansion}.
\end{Corollary}

For example,
\begin{gather*}
A_{0,0}^{\bP^1,+} = \sum_{d\geq 0} \frac{q^d}{d!}{\rm e}^{-q}
\sum_{\lambda\in \mathcal{P}_d} \frac{\bigl(\dim V^\lambda\bigr)^2}{d!}
\hat{\bs}_{(1)}(\lambda)\\ \hphantom{A_{0,0}^{\bP^1,+}}{}
 = {\rm e}^{-q}\biggl[0 + 1\cdot q + \biggl(\biggl(\frac{1}{2!}\biggr)^2 \cdot 2
+ \biggl(\frac{1}{2!}\biggr)^2 \cdot 2\biggr) q^2 \\ \hphantom{A_{0,0}^{\bP^1,+}={\rm e}^{-q}\biggl[}{}
 + \biggl( \biggl(\frac{1}{3!}\biggr)^2 \cdot 3
+ \biggl(\frac{2}{3!}\biggr)^2\cdot 3
+ \biggl(\frac{1}{3!}\biggr)^2 \cdot 3\biggr) q^3 + \cdots\biggr]
 = q,
\\
A_{1,0}^{\bP^1,+} = \sum_{d\geq 0} \frac{q^d}{d!}{\rm e}^{-q}
\sum_{\lambda\in \mathcal{P}_d} \frac{\bigl(\dim V^\lambda\bigr)^2}{d!}
\hat{\bs}_{(2)}(\lambda)\\ \hphantom{A_{1,0}^{\bP^1,+}}{}
 = \sum_{d\geq 0} \frac{q^d}{d!}{\rm e}^{-q}
\sum_{\lambda\in \mathcal{P}_d} \frac{\bigl(\dim V^\lambda\bigr)^2}{d!} \biggl(\frac{1}{2}\hat{\bp}_{(2)}(\lambda)
+ \frac{1}{2} \hat{\bp}_{(1)}(\lambda)^2\biggr) \\ \hphantom{A_{1,0}^{\bP^1,+}}{}
 = \frac{1}{2} {\rm e}^{-q}\biggl[0+ 0 \cdot q + \biggl( \biggl(\frac{1}{2!}\biggr)^2 \cdot 2
+ \biggl(\frac{1}{2!}\biggr)^2 \cdot (-2) \biggr) q^2 \\ \hphantom{A_{1,0}^{\bP^1,+}=\frac{1}{2} {\rm e}^{-q}\biggl[}{}
 + \biggl( \biggl(\frac{1}{3!}\biggr)^2 \cdot 6
+ \biggl(\frac{2}{3!}\biggr)^2\cdot 0
+ \biggl(\frac{1}{3!}\biggr)^2 \cdot (-6) \biggr) q^3 + \cdots \\ \hphantom{A_{1,0}^{\bP^1,+}=\frac{1}{2} {\rm e}^{-q}\biggl[}{}
 + 0^2q + 1^2q
+ \biggl( \biggl(\frac{1}{2!}\biggr)^2 \cdot 2^2
+ \biggl(\frac{1}{2!}\biggr)^2 \cdot 2^2 \biggr) q^2 \\ \hphantom{A_{1,0}^{\bP^1,+}=\frac{1}{2} {\rm e}^{-q}\biggl[}{}
 + \biggl( \biggl(\frac{1}{3!}\biggr)^2 \cdot 3^2
+ \biggl(\frac{2}{3!}\biggr)^2 \cdot 3^2
+ \biggl(\frac{1}{3!}\biggr)^2 \cdot 3^2\biggr) q^3 + \cdots\biggr]
 = \frac{q}{2} + \frac{q^2}{2}.
\end{gather*}
These match with the right-hand side of~\eqref{id21004}. 
The Plancherel averages of $\hat{\bp}_\mu$ and $\hat{\bs}_\mu$
are again polynomials, but they seem to be simpler than
the averages of $\bp_\mu$ and $\bs_\mu$.
For example,
\begin{alignat*}{4}
&\corr{\hat{\bp}_1}_d = d,\qquad&&
\corr{\hat{\bp}_2}_d = 0,\qquad&&
\bigl\langle\hat{\bp}\bigr\rangle_1^2 = d^2,& \\
&\corr{\hat{\bp}_3}_d = \frac{3}{2}d^2 - \frac{5d}{4},\qquad&&
\corr{\hat{\bp}_2\hat{\bp}_1}_d = 0,\qquad&&
\bigl\langle\hat{\bp}_1^3\bigr\rangle_d = d^3, &
\end{alignat*}
and from these we get
\begin{gather*}
\bigl\langle\hat{\bs}_{(1)}\bigr\rangle_d = d, \\
\bigl\langle\hat{\bs}_{(2)}\bigr\rangle_d = \bigl\langle\hat{\bs}_{(1^2)}\bigr\rangle
= \frac{d^2}{2},\\
\bigl\langle\hat{\bs}_{(3)}\bigr\rangle_d = \bigl\langle\hat{\bs}_{(1^3)}\bigr\rangle
= \frac{d^3}{6} +\frac{1}{2}d^2-\frac{5}{12}d,
\\
\bigl\langle\hat{\bs}_{(21)}\bigr\rangle_d =
\frac{d^3}{3} - \frac{1}{2}d^2
+\frac{5}{12}d.
\end{gather*}

\subsection*{Acknowledgements}
One of the authors D.Y. is grateful to
Marco Bertola, Boris Dubrovin and Youjin Zhang for their advice. We thank
Don Zagier for several insightful and helpful comments. We also
thank the anonymous referees for constructive comments that help to improve the presentation.
The work is supported by NSFC (No.~12371254, No.~11890662) and
CAS No. YSBR-032.

\pdfbookmark[1]{References}{ref}
\LastPageEnding

\end{document}